\documentclass[amsmath,secnumarabic,floatfix,amssymb,nofootinbib,nobibnotes,letterpaper,11pt,tightenlines]{revtex4}

\usepackage{times}

\usepackage{geometry}
\usepackage{amssymb}
\usepackage{latexsym, amsmath, amscd,amsthm}
\usepackage{graphicx}
\usepackage[percent]{overpic}
\usepackage{units}
\usepackage{hyperref}
\PassOptionsToPackage{caption=false}{subfig}
\usepackage[lofdepth]{subfig}
\usepackage{booktabs}

\usepackage{clrscode}

\newtheorem{theorem}{Theorem}

\newtheorem{lemma}[theorem]{Lemma}
\newtheorem{proposition}[theorem]{Proposition}
\newtheorem{corollary}[theorem]{Corollary}

\theoremstyle{definition}
\newtheorem{definition}[theorem]{Definition}

\newtheorem{conjecture}[theorem]{Conjecture}

\def\thm#1{Theorem~\ref{thm:#1}}

\def\figr#1{Figure~\ref{fig:#1}}

\newcommand{\R}{\mathbb{R}}

\newcommand{\norm}[1]{\left\| #1 \right\|}

\newcommand{\Ad}{\operatorname{Ad}}
\newcommand{\Adx}{\Ad_{\pmb{x}}}
\newcommand{\Adxy}{\Adx(\vec{y})}

\newcommand{\gradAd}{\nabla\!\Ad}
\newcommand{\gradAdx}{\nabla\!\Adx}
\newcommand{\gradAdxy}{\nabla\!\Adxy}

\newcommand{\HAdx}{\mathcal{H}\Adx}
\newcommand{\HAdxy}{\mathcal{H}\Adxy}

\newcommand{\lmin}{\lambda_{\operatorname{min}}}
\newcommand{\lmax}{\lambda_{\operatorname{max}}}

\newcommand{\xxt}{\hat{x}_i \hat{x}_i^T}
\newcommand{\wxxt}{\omega_i \hat{x}_i \hat{x}_i^T}

\newcommand{\Ed}{\operatorname{Ed}}
\newcommand{\dx}{\,\mathrm{d}x}

\newcommand{\Pol}{\operatorname{Pol}}

\newcommand{\Arm}{\operatorname{Arm}}

\newcommand{\dVol}{\thinspace\operatorname{dVol}}
\newcommand{\dArea}{\thinspace\operatorname{dArea}}
\newcommand{\Vol}{\operatorname{Vol}}

\newcommand{\gm}{\vec{\mu}}
\newcommand{\gmc}{\operatorname{gmc}}

\newcommand{\Var}{\operatorname{Var}}

\let\mgp=\marginpar \marginparwidth18mm \marginparsep1mm
\def\marginpar#1{\mgp{\raggedright\tiny #1}}

\let\lbl=\label
\def\label#1{\lbl{#1}\ifinner\else\marginpar{\ref{#1} #1}\ignorespaces\fi}

\begin{document}
\title[]{Open and closed random walks with fixed edgelengths in $\R^d$}
\author{Jason Cantarella}
\altaffiliation{Mathematics Department, University of Georgia, Athens GA}
\noaffiliation
\author{Kyle Chapman}
\altaffiliation{Mathematics Department, University of Georgia, Athens GA}
\noaffiliation
\author{Philipp Reiter}
\altaffiliation{Mathematics Department, University of Georgia, Athens GA}
\noaffiliation
\author{Clayton Shonkwiler}
\altaffiliation{Department of Mathematics, Colorado State University, Fort Collins CO}
\noaffiliation

\keywords{Fermat-Weber problem, geometric median, random polygon, random knot, concentration of measure, Nakagami distribution, Bernstein inequality, Haldane's spatial median}

\begin{abstract} 
In this paper, we consider fixed edgelength $n$-step random walks in $\R^d$. We give an explicit construction for the closest closed equilateral random walk to almost any open equilateral random walk based on the geometric median, providing a natural map from open polygons to closed polygons of the same edgelength. Using this, we first prove that a natural reconfiguration distance to closure converges in distribution to a Nakagami$(\nicefrac{d}{2},\nicefrac{d}{d-1})$ random variable as $n \rightarrow \infty$. We then strengthen this to an explicit probabilistic bound on the distance to closure for a random $n$-gon in any dimension with any collection of fixed edgelengths $w_i$. Numerical evidence supports the conjecture that our closure map pushes forward the natural probability measure on open polygons to something very close to the natural probability measure on closed polygons; if this is so, we can draw some conclusions about the frequency of local knots in closed polygons of fixed edgelength.
\end{abstract}
\date{\today}
\maketitle

\section{Introduction}
Random walks in space with fixed edgelengths have been of interest to statistical physicists and chemists since Lord Rayleigh's day. These walks model polymers in solution (at least under $\theta$-solvent conditions)~\cite{Rayleigh:1919do,hughes1995random,FloryPaulJ1969Smoc} and are similarly interesting in computational geometry and mathematics as a space of ``linkages''~\cite{Demaine:2007jh,MR2004a:14059}.
While 2- and 3-dimensional walks are the most relevant to this case, high-dimensional random walks often shed light on the lower dimensional situation~\cite{Rudnick:1987jn}.

In this paper, we will consider the relationship between open and closed random walks of fixed edgelengths. We will provide an explicit algorithm for finding the nearest closed polygon with given edgelengths to almost any collection of edge directions, and use our construction to provide tail bounds on the fraction of polygon space within a fixed distance of the closed polygons in any dimension. Our results will be strongest for equilateral polygons, but provide explicit bounds for any collection of edgelengths.  

To establish notation, we describe random walks in $\R^d$ with (fixed) positive edgelengths $w_i$ by their \emph{edge clouds} $(w_1, \hat{x}_1),\ldots , (w_n,\hat{x}_n)$ where $\hat{x}_i \in S^{d-1}$ is the direction of the $i$th edge. The space of polygonal arms $\Arm(n,d,w)$ is topologically equivalent to $(S^{d-1})^n$. If we let $\omega_i = \frac{w_i}{\sum w_i}$ be the relative edgelengths, then we can define the submanifold $\{ \pmb{x} : \sum \omega_i \hat{x}_i = \vec{0} \}$ of closed polygons $\Pol(n,d,w)$.

Using Bernstein's inequality~(e.g. \cite{Dubhashi:2009ho}), there is an easy concentration inequality which suggests that the endpoints of random arms are close together. For equilateral polygons in $\R^3$, this takes the simple form
\begin{theorem}  
If $\pmb{x}$ is chosen randomly in $\Arm(n,3,1)$ with edges $\hat{x}_1,\ldots , \hat{x}_n$,
\label{thm:sum}
\begin{equation*}
\mathcal{P}\left( \frac{1}{n} \left\| \sum \hat{x}_i \right\| < t \right) \geq 1 - 3 \, 
e^{-n t^2 \cdot \frac{3}{6 + 2\sqrt{3} t}}.
\end{equation*}
\end{theorem}
That is, the center of mass of a random edge cloud is very likely to be close to the origin.
%Further, the concentration is at rate $n t^2$ for any fixed $d$. This rate is exactly what we would expect from the central limit theorem: it says that the probability that an $n$-edge arm has failure to close $\|\sum \hat{x_i} \| > n^{1/2 + \epsilon}$ goes to zero exponentially fast for any $\epsilon > 0$. 
We can clearly close a random polygon in $\Arm(n,3,1)$ by subtracting the (small) $\frac{1}{n} \sum \hat{x}_i$ from each edge. That closed polygon is clearly near the original arm, but it is no longer equilateral. This raises the question of whether we can generally close a polygon in $\Arm(n,3,1)$ (or $\Arm(n,d,w)$) while preserving edgelengths and changing the polygon only a small amount. This question is the focus of this paper. 

Given $\pmb{x}$ and $\pmb{y}$ in $\Arm(n,d,w)$, we view both as vectors in $\R^{dn}$ and measure the distance between them accordingly. We call this the~\emph{chordal} distance because it does not measure the arc on the spheres of radius $w_i$ for each pair of edges, but rather measures the straight line distance between edge vectors. 

Our first main result is Proposition~\ref{prop:dchordal asymptotics}, which shows that the chordal distance between a random $\pmb{x} \in \Arm(n,d,1)$ and the nearest $\pmb{y} \in \Pol(n,d,1)$ converges in distribution to a Nakagami-$(\nicefrac{d}{2},\nicefrac{d}{d-1})$ random variable with PDF proportional to $x^{d-1} e^{-\frac{d-1}{2} x^2}$ as $n \rightarrow \infty$.

Our second main result is a general probabilistic bound on the chordal distance to closure for random polygons in any $\Arm(n,d,w)$. For equilateral polygons in $\R^3$, our main theorem (Corollary~\ref{cor:chordal concentration}) takes the very simple form
\begin{equation*}
\mathcal{P}\left(d_\text{chordal}(\pmb{x},\Pol(n,3,1)) < t \right) \geq 1 - 6 \exp\left( \nicefrac{-t^2}{4} \right)
\end{equation*}
for $t < \frac{\sqrt{n}}{200 \sqrt{2}}$. 

Here is a broad overview of our arguments. Given a polygon $\pmb{x}$ in $\Arm(n,d,w)$, we will provide an explicit construction for a nearby closed polygon in $\Pol(n,d,w)$, which we call the \emph{geometric median closure} of $\pmb{x}$ (denoted $\gmc(\pmb{x})$). It will be clear how to construct the geodesic in $\Arm(n,d,w)$ from $\pmb{x}$ to $\gmc(\pmb{x})$. For equilateral polygons, we show $\gmc(\pmb{x})$  is the closest closed polygon to $\pmb{x}$ in chordal distance (Theorem~\ref{thm:gmc is closest closed}).

The distance between $\pmb{x}$ and $\gmc(\pmb{x})$ depends on the norm $\| \gm \|$ of the geometric median (or Fermat-Weber point) of the edge cloud (Proposition~\ref{prop:distance bound}). For equilateral polygons, we will be able to leverage existing results of Niemiro~\cite{Niemiro:1992ez} to find the asymptotic distribution of the geometric median of a random point cloud (Proposition~\ref{prop:geometric median asymptotics}). Combining this with the matrix Chernoff inequalities proves our first main result (Proposition~\ref{prop:dchordal asymptotics}).

The second main result follows from a concentration inequality for a random polygon in any $\Arm(n,d,w)$, which bounds the probability of a large $\| \gm \|$ in terms of $n$, $d$, and $w$. This concentration result  (Theorem~\ref{thm:main}) follows from parallel uses of the scalar and matrix Bernstein inequalities to control the expected properties of a random edge cloud, together with the definition of the geometric median as the minimum of a convex function. 

Last, we will observe that the pushforward measure on closed polygons obtained by closing random open polygons appears to converge rapidly to the uniform distribution on closed polygons (Conjecture~\ref{conj:pushforward}). Since these closures involve only very small motions of any part of the polygon, local features (such as small knots) should be preserved -- it would follow (Conjecture~\ref{conj:local knotting}) that the rate of production of local knots in open and closed arcs should be almost the same. 

\section{Constructing a nearby closed polygon}

As mentioned above, we view $n$-edge polygons (up to translation) in $\R^d$ as collections of edge vectors $\vec{x}_i \in \R^d$.\footnote{Throughout this paper, we use boldface to indicate elements of $\R^{dn}$, which we usually think of as vectors of edge vectors. We use a superscript arrow -- as in $\vec{x}_i$ -- to denote an arbitrary element of $\R^d$, though any such vector which is definitionally a unit vector we mark with a hat rather than an arrow.} The vertices are obtained by summing the $\vec{x}_i$ from an arbitrary basepoint. In this section of the paper, we will assume only that the lengths of the edges of the polygon are fixed to some arbitrary $w_i = \|\vec{x}_i\|$. We will think of these fixed edgelength polygons in two ways:
\begin{itemize}
\item as a weighted point cloud on the unit sphere $S^{d-1} \subset \R^d$ where the points are denoted $\hat{x}_i = \vec{x}_i/\|\vec{x}_i\|$ and the weights are the $w_i$. We will call $(w_i,\hat{x}_i)$ the \emph{edge cloud} of the polygon.
\item as a point $\pmb{x} \in \prod S^{d-1}(w_i) \subset (\R^d)^n = \R^{dn}$ (where $S^{d-1}(r)$ is the sphere of radius $r$). We will call $\pmb{x}$ the \emph{vector of edges} of the polygon.
\end{itemize}
The space of these polygons will be denoted $\Arm(n,d,w) = \prod  S^{d-1}(w_i)$. Within this space, there is a submanifold $\Pol(n,d,w)$ of closed polygons defined by the condition $\sum w_i \hat{x}_i = \vec{0}$. (Equivalently, $\pmb{x}$ is closed if it lies in the codimension $d$ subspace of $\R^{dn}$ normal to the $\pmb{n}^j = (\hat{e}_j,\ldots , \hat{e}_j)$, where $\hat{e}_1, \ldots , \hat{e}_d$ is the standard basis in $\R^d$.) Both $\Arm(n,d,w)$ and $\Pol(n,d,w)$ are Riemannian manifolds with standard metrics, but it will be useful to use two additional metrics as well:

\begin{definition}
The \emph{chordal} metric on $\Arm(n,d,w)$ is given by 
\begin{equation*}
d_\text{chordal}(\pmb{x},\pmb{y}) = \|\pmb{x} - \pmb{y}\|_{\R^{dn}} 
= \left( \sum \|w_i \hat{x}_i - w_i \hat{y}_i\|_{\R^d}^2 \right)^{\nicefrac{1}{2}}.
\end{equation*}
The \emph{max-angular} metric on $\Arm(n,d,w)$ is given by
\begin{equation*}
d_\text{max-angular}(\pmb{x},\pmb{y}) = \max_i \angle(\vec{x}_i,\vec{y}_i).
\end{equation*}
\end{definition}

We now make an important definition:

\begin{definition}
A \emph{geometric median} (also known as a \emph{Fermat-Weber point}) of an edge cloud 
$(w_i,\hat{x}_i)$ is any point $\gm$ which minimizes the \emph{weighted average distance function} $\Adxy$ given by 
\begin{equation*}
\Adxy = \sum_i \omega_i \| \hat{x}_i - \vec{y} \|.
\end{equation*}
where $\omega_i = \nicefrac{w_i}{\sum w_i}$.
To clarify notation, we will only use $\gm$ for points which are a geometric median of a weighted point cloud; the point cloud will be clear from the context.
\label{def:gm}
\end{definition}
This is a very old construction with a beautiful theory around it; see the nice review in \cite{Hamacher:2002vp}. We note that the geometric median differs from the center of mass (or geometric \emph{mean}) of the points, which minimizes the weighted average of the \emph{squared} distances between $\vec{y}$ and the $\hat{x}_i$ and that the geometric median is unique unless
the points are all colinear and the geometric median is not one of the points.  

This section is devoted to analyzing the following construction:

\begin{definition}
Suppose $\pmb{x}$ is a polygon and $\gm$ is a geometric median of its edge cloud $(w_i, \hat{x}_i)$ which is not one of the $\hat{x}_i$. The \emph{geometric median closure} $\gmc(\pmb{x})$ of $\pmb{x}$ is the polygon whose edge cloud has the same weights and edge directions obtained by recentering the $\hat{x}_i$ on $\gm$ and renormalizing: 
$\gmc(\pmb{x})$ has edge cloud $\left( w_i, \frac{\hat{x}_i - \gm}{\left| \hat{x}_i - \gm \right|} \right)$. 

If every geometric median of $(w_i,\hat{x}_i)$ is one of the $\hat{x}_i$, $\gmc(\pmb{x})$ is not defined. If $\gmc(\pmb{x})$ is defined, we say that $\pmb{x}$ is \emph{median-closeable}.
\label{def:gmc}
\end{definition}

Of course, we need to justify our choice of name by proving that $\gmc(\pmb{x})$ is closed. The key observation is the following Lemma, which follows by direct computation:
\begin{lemma}
The function $\Adxy$ is a convex function of $\vec{y}$. The gradient is given by
\begin{equation*}
\gradAdxy = \sum \omega_i \frac{\vec{y} - \hat{x}_i}{\|\vec{y} - \hat{x}_i\|}.
\end{equation*}
The Hessian of $\Adxy$ is given by
\begin{equation*}
\HAdxy = \left(\sum_i \frac{\omega_i}{\norm{\hat{x}_i - \vec{y}}}\right) I_d - 
\left(\sum_i \frac{\omega_i}{\norm{\hat{x}_i - \vec{y}}^3} (\vec{y} - \hat{x}_i)(\vec{y} - \hat{x}_i)^T \right).
\end{equation*}
\label{lem:gradad and had}
\end{lemma}

\begin{proposition}
If $\pmb{x}$ is median-closeable, $\gmc(\pmb{x})$ is a unique closed polygon with edgelengths $w_i$.
\label{prop:gmc is closed}
\end{proposition}

\begin{proof}
The proof follows from assembling several standard facts about the geometric median. These are in \cite{Anonymous:1UuVxm-1}, but are easily checked by hand. 

As it is a sum of convex functions, the average distance function $\Adx$ is convex. Away from the $\hat{x}_i$, it is differentiable. If the points $\hat{x}_i$ are not colinear, $\Adx$ is strictly convex and $\gm$ is unique. If the points $\hat{x}_i$ are colinear, either the geometric median is one of the $\hat{x}_i$ or the set of geometric medians consists of the interval between two $\hat{x}_i$.

Any geometric median which is not one of the $\hat{x}_i$ must be a critical point of the average distance function. For any such $\vec{\mu}$, using Lemma~\ref{lem:gradad and had}, 
\begin{equation}
\sum_i  w_i \frac{\hat{x}_i - \gm}{\left|\hat{x}_i - \gm\right|} = -\left(\sum w_i\right) \gradAdx(\gm) = 0.
\label{eq:grad total distance}
\end{equation}
This implies that $\gmc(\pmb{x})$ is closed. 

If $\gm$ is unique, then $\gmc(\pmb{x})$ is obviously unique. If $\gm$ is not unique, the $\hat{x}_i$ are colinear, and $\gm$ is on the line segment between two of the $\hat{x}_i$. In this case, it is not hard to see that~\eqref{eq:grad total distance} implies that the edges of $\gmc(\pmb{x})$ are two antipodal groups of points on $S^{k-1}$, each containing $n/2$ points, regardless of where we take $\gm$ on the segment.
\end{proof}

Our next goal is to prove an optimality property for the geometric median closure. We will start by proving a more general fact about recentering and renormalizing:

\begin{proposition}
Suppose that $\hat{x}_i$ is any point cloud in $(S^{d-1})^n$, and $\vec{p} \in \R^d$ is not one of the $\hat{x}_i$. Given any set of weights $w_i$, we let $r(\pmb{x};\vec{p},w)$ denote the renormalized, recentered, and reweighted point cloud, and $\vec{s}$ denote its weighted sum:
\begin{equation*}
r(\pmb{x};\vec{p},w) := 
\left( w_i, 
\frac{\hat{x}_i - \vec{p}}{\left| \hat{x}_i - \vec{p}\right|} 
\right) 
\quad \text{and} \quad 
\vec{s} := \sum_i w_i \frac{\hat{x}_i - \vec{p}}{\left| \hat{x}_i - \vec{p}\right|}.
\end{equation*}
If $\pmb{x}=(\hat{x}_1,\ldots , \hat{x}_n)$ and $\pmb{r}$ is the vector of edges corresponding to the edge cloud $(w_i,\hat{r}_i)$, then 
\begin{equation*}
\pmb{r} = \operatorname{argmin}_{\sum w_i \hat{y}_i = \vec{s}} \|\pmb{y} - \pmb{x}\|,
\end{equation*}
that is, $\pmb{r}$ is the closest vector of edges to $\pmb{x}$ (in $\R^{dn}$) with edge weights $w_i$ and vector sum $\vec{s}$.
\label{prop:recentering and renormalizing is closest}
\end{proposition}

%Note that when we say ``extrinsically closest'' above, we mean that we are measuring distance between the edge vectors $\vec{x}$ and $r(\vec{x};\vec{p})$ in the ambient Euclidean space $\R^{kn}$ instead of measuring the distance between the edge clouds instrinsically in the Riemannian manifold $(S^{k-1})^n$.

\begin{proof}
Suppose that $(w_i,\hat{y}_i)$ is a point cloud with the same weights which also has $\sum_i w_i \hat{y}_i = \vec{s}$ and $\pmb{y}$ is the corresponding vector of edges in $\R^{dn}$. Let $\pmb{v} = \pmb{y} - \pmb{r}$. Since $\sum_i \vec{y}_i = \sum_i \vec{r}_i$, we know $\sum \vec{v}_i = \vec{0}$. 

Remembering that $\|\vec{y}_i\| = w_i = \|\vec{r}_i\|$, we compute 
\begin{equation*}
\left< \vec{v}_i, \vec{r}_i \right> = 
\left< \vec{y}_i - \vec{r}_i, \vec{r}_i \right> = \left< \vec{y}_i, \vec{r}_i \right> - w_i^2 \leq 0.
\end{equation*} 
Since $\vec{r}_i$ is a positive scalar multiple of $\hat{x}_i - \vec{p}$, this implies that $\left< \vec{v}_i,\hat{x}_i - \vec{p} \right> \leq 0$, and so we have $\left<\vec{v}_i,\hat{x}_i\right> \leq \left< \vec{v}_i, \vec{p} \right>$. 
Since $\sum \vec{v}_i = \vec{0}$, we see
\begin{equation*}
\left< \pmb{v}, \pmb{x} \right> = \sum_i \left< \vec{v}_i, \hat{x}_i \right> \leq \sum_i \left< \vec{v}_i, \vec{p} \right> 
=  \left< \sum_i \vec{v}_i, \vec{p} \right> = 0,
\end{equation*}
or that $-\left< \pmb{x},\pmb{v} \right> \geq 0$. Using the facts $\left< \pmb{y}, \pmb{y} \right> = \sum w_i^2 = \left< \pmb{r}, \pmb{r} \right>$ and $\pmb{y} = \pmb{r} + \pmb{v}$, 
\begin{align*}
\left< \pmb{x} - \pmb{y},\pmb{x} - \pmb{y} \right> &= 
\left< \pmb{x}, \pmb{x} \right> + \left< \pmb{y}, \pmb{y} \right> - 2 \left< \pmb{x}, \pmb{y} \right> \\
&= \left< \pmb{x}, \pmb{x} \right> + \left< \pmb{r},\pmb{r} \right> - 2 \left< \pmb{x}, \pmb{r} \right> - 2 \left<\pmb{x},\pmb{v} \right>  \\
&= \left<\pmb{x} - \pmb{r}, \pmb{x} -\pmb{r}  \right> - 2 \left< \pmb{x}, \pmb{v} \right>
\end{align*}
so $\left\| \pmb{x} - \pmb{y} \right\| \geq \left\| \pmb{x} - \pmb{r} \right\|$, as claimed.
\end{proof}

Combining Propositions~\ref{prop:gmc is closed} and~\ref{prop:recentering and renormalizing is closest} with Definition~\ref{def:gmc}, we have

\begin{theorem}
If $\pmb{x}$ is a median-closeable equilateral polygon, $\gmc(\pmb{x})$ is the closed equilateral polygon closest to $\pmb{x}$ in the chordal metric. 
\label{thm:gmc is closest closed}
\end{theorem}

\noindent\textbf{Remarks.} This construction may seem unexpected, but it has deep roots. In~\cite{Kapovich:1996p2605}, Kapovich and Millson provide an analogous closure construction which associates a unique closed equilateral polygon to any equilateral polygon where no more than half the edge vectors coincide by viewing the unit ball as the Poincar\'e ball model of hyperbolic space and (essentially) recentering and renormalizing in hyperbolic geometry around a point called the ``conformal median'' (see~\cite{Douady:1986go}) which is in many ways parallel to the geometric median. This is an example of a ``Geometric Invariant Theory'' (or GIT) quotient: see~\cite{Howard:2008uy}. These ideas inspired our work above: we did not adopt them entirely only because working in hyperbolic geometry makes the whole endeavor seem much more abstract and because we have not managed to prove an optimality property for their construction analogous to Theorem~\ref{thm:gmc is closest closed}. 

\section{Asymptotics of the geometric median and the distance to closure}

Now that we've established the connection between the geometric median and closure, we will establish some facts about the large-$n$ behavior of the geometric median. Since the geometric median is a symmetric estimator of a large number of i.i.d. random variables, it seems natural to expect that the distribution of $\mu$ should converge to a multivariate normal, even though the classical central limit theorem doesn't apply. In fact, this is true:

\begin{proposition}
Let $n$ points $\hat{x}_i$ be sampled independently and uniformly on $S^{d-1}$, with geometric median $\gm$. The random variable $\sqrt{n}\,\gm$ converges in distribution to $\mathcal{N}\left(\vec{0},\frac{d}{(d-1)^2}I_d\right)$ as $n \rightarrow \infty$. This implies that $\norm{\sqrt{n}\, \gm}$ converges in distribution to a Nakagami$\left(\frac{d}{2},\left(\frac{d}{d-1}\right)^2\right)$ random variable.
\label{prop:geometric median asymptotics}
\end{proposition}

\begin{figure}[t!]
	\centering
		\includegraphics[width=4in]{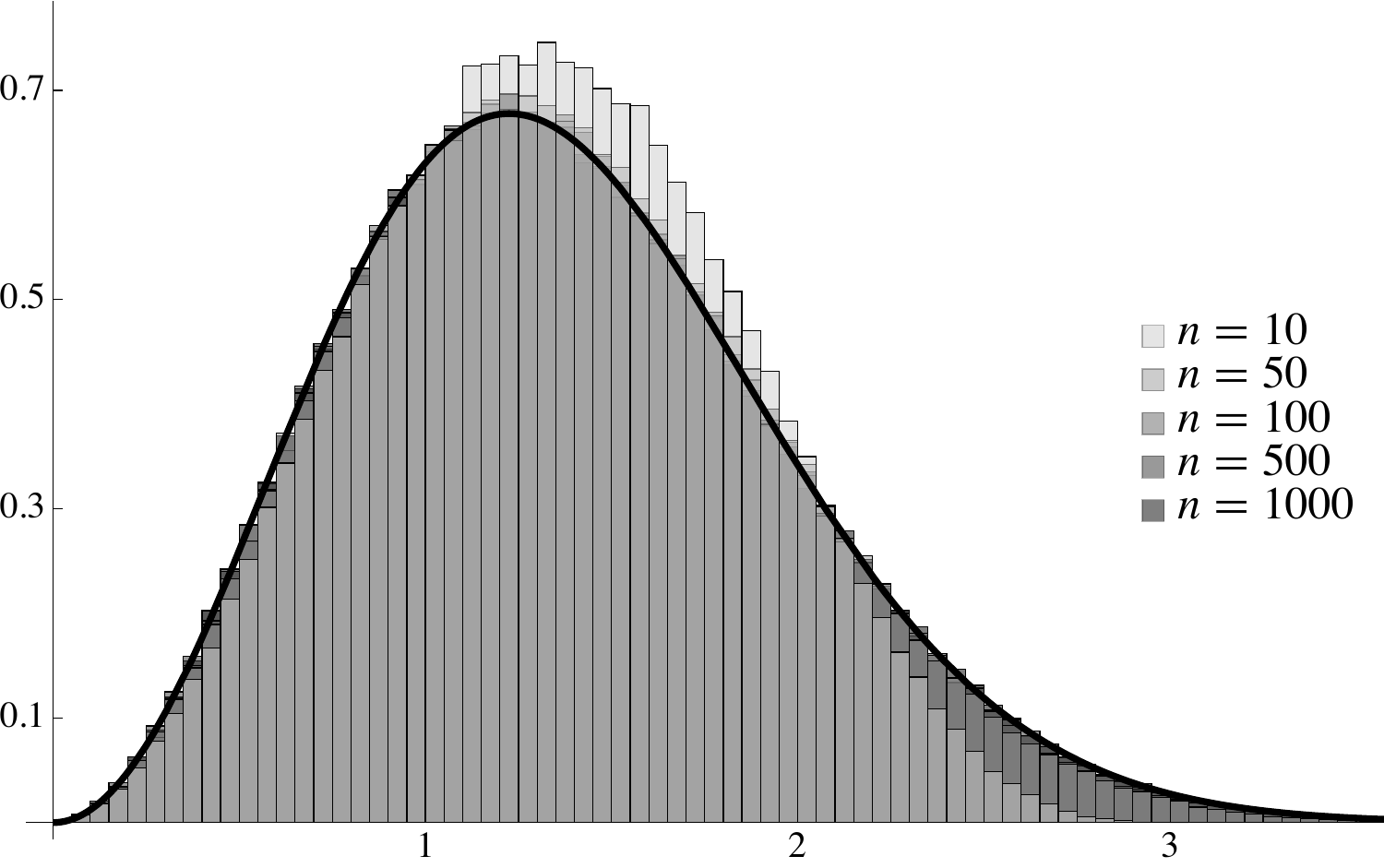}
	\caption{For various $n$, we generated 250,000 random elements of $\Arm(n,3,1)$ and computed $\sqrt{n}\,\|\gm\|$, where $\gm$ is the geometric median of the edge cloud. By Proposition~\ref{prop:geometric median asymptotics}, $\sqrt{n}\,\|\gm\|$ converges to a Nakagami$\left(\frac{3}{2},\frac{9}{4}\right)$ distribution, the pdf of which is the solid curve. Though we don't show it, the behavior in other dimensions is quite similar: by $n=50$ the density of the limiting Nakagami$\left(\frac{d}{2},\left(\frac{d}{d-1}\right)^2\right)$ distribution matches the histogram rather well.}
	\label{fig:convergence of geometric median}
\end{figure}

\begin{proof}
We start by defining $\Ed(\vec{y})$ to be the expected distance from $\vec{y} \in \R^d$ to the unit sphere; formally,
\[
	\Ed(\vec{y}) := \frac{1}{\Vol S^{d-1}} \int_{\hat{x} \in S^{d-1}} \|\hat{x} - \vec{y}\| \dVol_{S^{d-1}}.
\]
Observe that, by symmetry, the minimizer of $\Ed(\vec{y})$ is the origin. Now the geometric median of a finite collection of points $\hat{x}_i$ uniformly sampled from the sphere is the minimizer of the average distance $\Ad$ to the $\hat{x}_i$ (Definition~\ref{def:gm}). For a large number of points, we expect $\Ad$ to be close to $\Ed$ as a function, and hence that the minimizers of the functions should be nearby as well. 

In fact, Niemiro studied exactly this situation, showing\footnote{Under some technical hypotheses which are obviously satisfied in our case.}~(\cite[p. 1517]{Niemiro:1992ez}, cf. Haberman~\cite{Haberman:1989iq}) that
\begin{equation*}
\sqrt{n}\, \gm \overset{d}{\rightarrow} \mathcal{N}(\vec{0},\mathcal{H}^{-1} V \mathcal{H}^{-1})
\end{equation*}
where $V$ is the covariance matrix of a random point $\hat{x}$ on $S^{d-1}$ and $\mathcal{H}$ is the Hessian of $\Ed$, evaluated at the origin.

The off-diagonal elements of $V$ are zero by symmetry. Using cylindrical coordinates on $S^{d-1}$ with axis $\hat{e}_i$, the $i$th diagonal entry in the covariance matrix is computed by the integral
\[
\sigma_i^2 = \frac{1}{\sqrt{\pi}}\frac{\Gamma(\frac{d}{2})}{\Gamma(\frac{d - 1}{2})} \int_{-1}^1
 x^2 (1 - x^2)^{\frac{d - 3}{2}} \,\dx = \frac{1}{d}.
\]

We prove in the \href{appendix}{Appendix} (Proposition~\ref{prop:mister ed}) that the expected distance function $\Ed(\vec{y})$ is given as a function of $r = \|\vec{y}\|$ by
\[
\Ed(r) = \, _2F_1 \left( -\frac{1}{2}, \frac{1-d}{2}; \frac{d}{2}; r^2 \right).
\]
When $d$ is odd, the standard Taylor series representation of the hypergeometric function truncates, and $\Ed(r)$ is a polynomial in $r$. For example, when $d=3$ we have $\Ed(r) = 1+\nicefrac{r^2}{3}$. In turn, a straightforward computation shows that the Hessian of $\Ed$ evaluated at the origin is simply
\[
	\mathcal{H} = \mathcal{H} \Ed(\vec{0}) = \frac{d-1}{d} I_d,
\]
where $I_d$ is the $d \times d$ identity matrix. This completes the proof of the first statement. To get the second, we note that the norm of a Gaussian $\mathcal{N}(\vec{0},\sigma^2I_d)$ random variate is Nakagami$\left(\frac{d}{2},d \sigma^2\right)$-distributed.
\end{proof}

We now see that the geometric median is becoming asymptotically normal, and concentrating around the origin. We can use this to prove an asymptotic result for the distance to closure for equilateral polygons.

\begin{figure}[t!]
	\centering
		\includegraphics[width=4in]{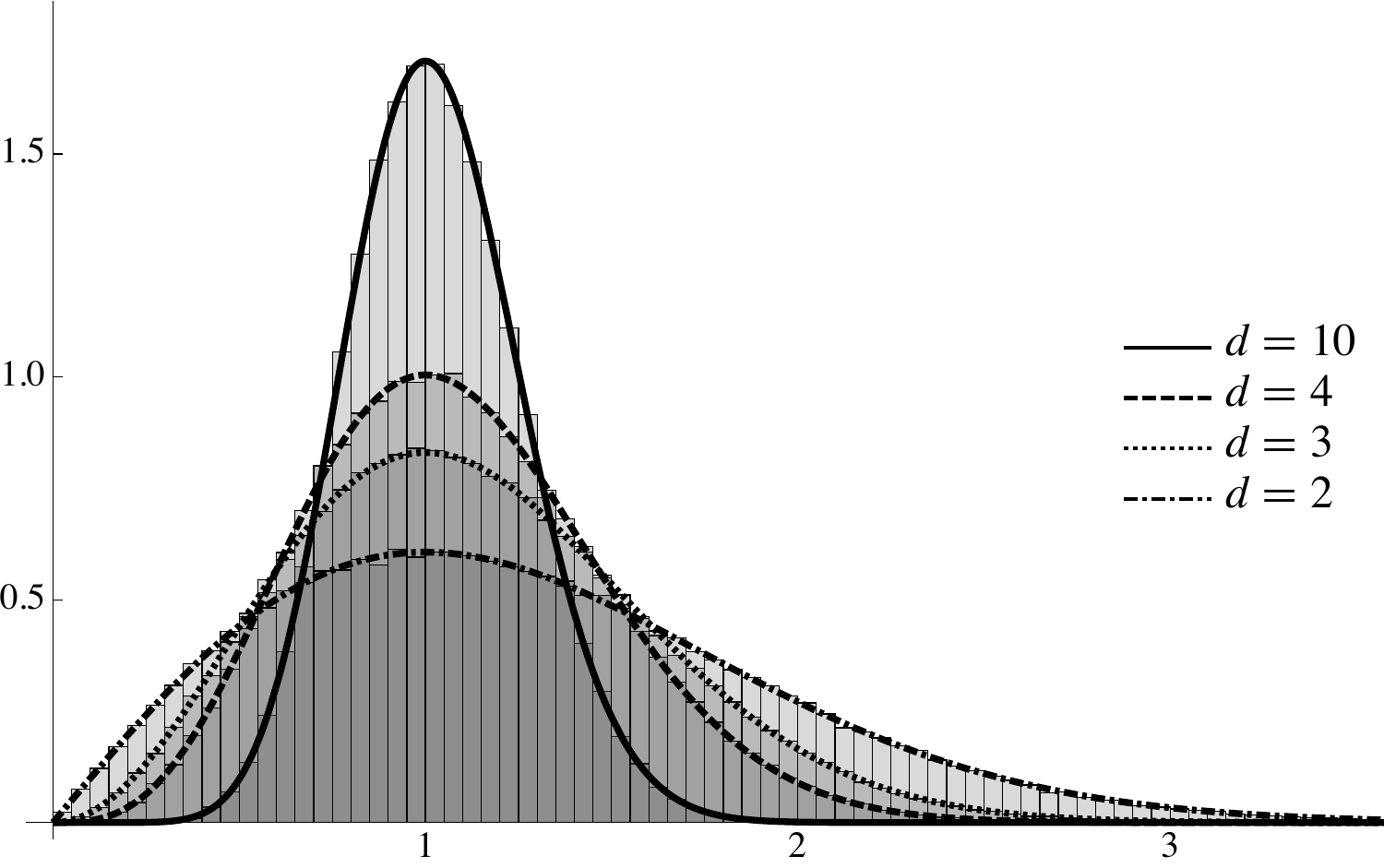}
	\caption{For $d=2,3,4,10$, we generated 250,000 random elements of $\Arm(1000,d,1)$ and computed their chordal distance to $\Pol(1000,d,1)$ using \thm{gmc is closest closed}. This plot shows the histograms of chordal distance to closure together with the densities of Nakagami$\left(\frac{d}{2},\frac{d}{d-1}\right)$ distributions.}
	\label{fig:chordal distance histograms}
\end{figure}

\begin{proposition}
For a random equilateral $n$-gon $\pmb{x}$ with edges $\hat{x}_i$ sampled independently and uniformly from $S^{d-1}$, the random variable $d_\text{chordal}(\pmb{x},\Pol(n,d,1))$ converges in distribution to a Nakagami$\left(\frac{d}{2},\frac{d}{d-1}\right)$ as $n \rightarrow \infty$. 
%That is,  
%\begin{equation*}
%\operatorname{PDF}(d_\text{chordal}(\pmb{x},\Pol(n,d,1))) \rightarrow \left(\frac{2^{1-\frac{d}{2}} (d-1)^{d/2}}{\Gamma
%   \left(\frac{d}{2}\right)}
%   \right)
%   x^{d-1}
%   e^{-\frac{d-1}{2} x^2}
%\end{equation*}
\label{prop:dchordal asymptotics}
\end{proposition}

\begin{proof}
We know from Theorem~\ref{thm:gmc is closest closed} that $d_\text{chordal}(\pmb{x},\Pol(n,d,1))$ is actually the chordal distance from $\pmb{x}$ to $\gmc(\pmb{x})$. To estimate this distance, we will make use of the recentering and renormalizing map $r(\pmb{x},\vec{p},1)$ from Proposition~\ref{prop:recentering and renormalizing is closest}.

When $\norm{\mu}$ is small, we can estimate
\begin{equation*}
\norm{\gmc(\pmb{x}) - \pmb{x}} = \norm{r(\pmb{x};\gm,1) - r(\pmb{x};\vec{0},1)} \sim \norm{\gm} \norm{D_{\hat{\mu}}r(\pmb{x};\vec{0},1)}
\end{equation*}
where $D_{\hat{\mu}} r(\pmb{x};\vec{0},1)$ is the derivative of $r(\pmb{x};\vec{v},1)$ with respect to the vector $\vec{v}$ in the direction of the unit vector $\hat{\mu} = \nicefrac{\gm}{\norm{\gm}}$ (while leaving the $\pmb{x}$ variables constant). 

Using the definition of $r(\pmb{x};\vec{p},1)$, a direct computation reveals that 
\begin{equation*}
\norm{D_{\hat{\mu}}r(\pmb{x};\vec{0},1)} = \left( n - \sum_i \left< \hat{x}_i, \hat{\mu} \right>^2 \right)^{\frac{1}{2}} = \sqrt{n} \left(1 - \frac{1}{n} \sum_i \left< \hat{x}_i, \hat{\mu} \right>^2 \right)^{\frac{1}{2}}.
\end{equation*}
Since $\hat{\mu}$ is a unit vector, the sum is the Rayleigh quotient for the matrix $X = \frac{1}{n} \sum_i \hat{x}_i \hat{x}_i^T$, and so obeys the estimates
\begin{equation*}
\lmin (X) \leq \frac{1}{n} \sum_i \left< \hat{x}_i, \hat{\mu} \right>^2 \leq \lmax (X)
\end{equation*}
Now $\lmin (X)$ and $\lmax (X)$ are also random variables depending on the $\hat{x}_i$, but we can use the matrix Chernoff inequalities~\cite[Remark 5.3]{Tropp:2012fb} to bound the probability that they are far from $\frac{1}{d}$. 

It's quite standard to prove that $\mathcal{E}(\hat{x}_i \hat{x}_i^T) = \frac{1}{d} I_d$, so $\mathcal{E}(X) = \frac{1}{d} I_d$. The matrix Chernoff inequalities then reduce to 
\begin{equation}
\mathcal{P}\left\{ \lmin (X) < (1 - \delta) \frac{1}{d} \right\} 
\leq d \left(\frac{e^{-\delta}}{(1-\delta)^{1-\delta}} \right)^{\frac{n}{d}}
\label{eq:lmin lower}
\end{equation}
and 
\begin{equation}
\mathcal{P}\left\{ \lmax (X) > (1+\delta) \frac{1}{d} \right\} \leq d \left(\frac{e^{\delta}}{(1+\delta)^{1+\delta}} \right)^{\frac{n}{d}}
\label{eq:lmax upper}
\end{equation}
For any $\delta > 0$, the quantities raised to $\frac{n}{d}$ are $< 1$, and so as $n \rightarrow \infty$ the probability that the bounds in~\eqref{eq:lmin lower} and~\eqref{eq:lmax upper} both hold $\rightarrow 1$. In turn, this means that for any fixed $\delta > 0$,
\begin{equation*}
\mathcal{P} \left\{ \left| \frac{1}{d} - \frac{1}{n} \sum_i \left< \hat{x}_i, \hat{\mu} \right>^2  \right| > \frac{\delta}{d} \right\} \rightarrow 0
\end{equation*}
and so the random variable $\frac{1}{n} \sum_i \left< \hat{x}_i, \hat{\mu} \right>^2$ converges in probability to $\frac{1}{d}$. By the continuous mapping theorem, this means that $\left(1 - \frac{1}{n} \sum_i \left< \hat{x}_i, \hat{\mu} \right>^2\right)^{1/2} \overset{p}{\rightarrow}\sqrt{\frac{d-1}{d}}$.

We can now rewrite the random variable $\norm{\gm} \norm{D_{\hat{\mu}}r(\pmb{x};\vec{0},1)}$ as the product of $\norm{\sqrt{n}\, \gm}$, which by Proposition~\ref{prop:geometric median asymptotics} converges in distribution to a Nakagami$\left(\frac{d}{2},\left(\frac{d}{d-1}\right)^2\right)$ random variable, and $\left(1 - \frac{1}{n} \sum_i \left< \hat{x}_i, \hat{\mu} \right>^2\right)^{1/2}$, which we have just proved converges in probability to the constant random variable $\sqrt{\frac{d-1}{d}}$. 

Using Slutsky's theorem and a little algebra, this implies that the product converges in distribution to a Nakagami $\left(\frac{d}{2},\frac{d}{d-1}\right)$ random variable, as claimed. 
\end{proof}

We have now learned something interesting: the distribution of chordal distances to closure should be converging to a distribution which doesn't depend on the number of edges! This is surprising because the diameter of $\Arm(n,d,1)$ is clearly $\Theta(\sqrt{n}) \rightarrow \infty$. This means that some arms might indeed be very far from closure -- but they are very rare. We will look for this feature in the more specific probability inequalities to come.

We can also see how fast the tail of the distribution of $d_\text{chordal}$ can be expected to decay. The survival function of the Nakagami distribution is an incomplete Gamma function. Using~\cite[8.10.1]{NIST:DLMF}, we can show that there is a constant $C(d) > 0$ so that if $x$ is Nakagami$\left(\frac{d}{2},\frac{d}{d-1}\right)$, then 
\begin{equation}
\mathcal{P} \left\{ x < t \right\} \geq 1 - C(d) \, t^{d-2} e^{-\frac{d-1}{2} t^2}.
\label{eq:precise tail bound}
\end{equation}

\section{Concentration inequalities for $\norm{\gm}$ and $d_\text{chordal}$}

We now know what to expect in the large-$n$ limit, at least for equilateral polygons:~\eqref{eq:precise tail bound} tells us that we should aim for a tail bound for $d_\text{chordal}$ which does not depend on $n$ and is proportional to $e^{-\alpha t^2}$ for some $\alpha < \frac{d-1}{2}$. We will get exactly such a bound in Corollary~\ref{cor:chordal concentration} at the end of the section. Our bounds will apply for finite $n$, and also apply to the non-equilateral case, where it is not even clear what the large-$n$ limit should mean. 

\subsection{A bound connecting $\norm{\gm}$ and $d_\text{chordal}$.}

To start with, we prove a hard bound on the relationship between the geometric median and our two measures of distance in polygon space. First, we note that our procedure of recentering and renormalizing changes each $\hat{x}_i$ by a controlled amount.

\begin{lemma}
If $\hat{x}_i \in S^{d-1}$ and $\vec{p} \in \R^d$ is any vector with $\| \vec{p} \| < 1$, then $\left\| \hat{x}_i - \frac{\hat{x}_i - \vec{p}}{\|\hat{x_i} - \vec{p}\|} \right\| \leq \sqrt{2} \| \vec{p} \|$ and $\angle (\hat{x}_i,\frac{\hat{x}_i - \vec{p}}{\|\hat{x}_i - \vec{p}\|}) \leq \arcsin \|\vec{p}\|< \frac{\pi}{2} \| \vec{p} \|$. 
\label{lem:distance bound}
\end{lemma}

% Note: The full writeup is in ground-truth.nb
\begin{proof}
%This is a calculus exercise-- with some algebra, it is straightforward to show
%\begin{equation*}
%\left\| \hat{x}_i - \frac{\hat{x}_i - \mu}{\|\hat{x_i} - \mu\|} \right\|^2 = \left< \hat{x}_i - \frac{\hat{x}_i - \mu}{\|\hat{x}_i - \mu\|}, 
%\hat{x}_i - \frac{\hat{x}_i - \mu}{\|\hat{x}_i - \mu\|} \right> = 
%2 + \frac{2 \left< \hat{x}_i, \mu \right> - 2}{\sqrt{1 + \|\mu\|^2 - 2 \left< \hat{x}_i,\mu \right>}}.
%\end{equation*}
%Without loss of generality, we may assume $\hat{x}_i = \vec{e}_1$, so $\left< \hat{x}_i, \mu \right> = \mu_1$, which is some scalar between $0$ and $\|\mu\|$. We then maximize this expression as a function of $\mu_1$. We compute
%\begin{equation*}
%\frac{d}{d\mu_1} \left(2 + \frac{2\mu_1 - 2}{\sqrt{1 + \|\mu\|^2 - 2 \mu_1}}\right) = \frac{2(\|\mu\|^2 - \mu_1)}{(1 + \|\mu\|^2 - 2 \mu_1)^{\nicefrac{3}{2}}}.
%\end{equation*}
%Plugging in $\mu_1 = \|\mu\|^2$ as the only critical point, we get the (sharp) bound 
This is a calculus exercise; it is straightforward to establish the (sharp) bound
\begin{equation*}
\left\| \hat{x}_i - \frac{\hat{x}_i - \vec{p}}{\|\hat{x_i} - \vec{p}\|} \right\| \leq \sqrt{2 - 2\sqrt{1 - \|\vec{p}\|^2}}.
\end{equation*}
Further, it is easy to check that the right-hand side is a convex function of $\| \vec{p} \|$ which is equal to $0$ when $\|\vec{p}\|=0$, and $\sqrt{2}$ when $\|\vec{p}\|=1$, so it is bounded above by the line $\sqrt{2}\, \| \vec{p} \|$. The angle bound is also straightforward.
\end{proof} 

We now can give a bound on the distance between a given $\pmb{x} \in \Arm(n,d,w)$ and $\Pol(n,d,w)$ in terms of the norm of the geometric median $\gm$ of the edge cloud $(\hat{x}_i,w_i)$. 

\begin{proposition}
If the edge cloud $(\hat{x}_i,w_i)$ has geometric median $\gm$ with $\|\gm\| < 1$, 
\begin{equation*}
d_{\text{chordal}}(\pmb{x},\Pol(n,d,w)) < \sqrt{2 \sum \omega_i^2}\, \|\gm\| 
\quad 
\text{and}
\quad 
d_{\text{max-angular}}(\pmb{x},\Pol(n,d,w)) < \arcsin \|\gm\|. 
%< \frac{\pi}{2} \| \mu \|.
\end{equation*}
\label{prop:distance bound}
\end{proposition}

\begin{proof}
Since $\|\gm\| < 1$, our polygon is median-closeable and $\gmc(\pmb{x})$ is a closed polygon with edge cloud $\left(w_i, \frac{\hat{x}_i - \gm}{\|\hat{x}_i - \gm\|}\right)$. Lemma~\ref{lem:distance bound} immediately yields the bound on $d_{\text{max-angular}}$; to get the chordal distance bound, we write
\begin{equation*}
\sqrt{ \sum \left\| w_i \hat{x}_i - w_i \frac{\hat{x}_i - \gm}{\|\hat{x}_i - \gm\|} \right\|^2 }  \leq
\sqrt{ \sum w_i^2 \cdot 2 \| \gm \|^2 }. %= \left(\sqrt{\sum w_i^2}\right) \sqrt{2} \, \|\mu\|.
\end{equation*}
\end{proof}

\subsection{Strategy for the tail bound}

To derive our explicit tail bound on the norm of the geometric median, our strategy is as follows. First, we will prove two probabilistic bounds: an upper bound on $\|\gradAdx(\vec{0})\|$ and a positive lower bound on $\lmin(\HAdx(\vec{0}))$. These will come from scalar and matrix versions of Bernstein's inequality.

If we restrict $\Adx$ to a scalar function $\Adx(z)$ on a ray from the origin, these bounds yield an upper bound on $|\Adx'(0)|$ and a lower bound on $\Adx''(0)$. We will get a uniform lower bound on $\Adx''(z)$ for $z \in [0,\nicefrac{1}{50}]$ by showing that, $\Adx''(z) \geq \Adx''(0) - 7 z$ on this interval. We prove this using the special structure of $\HAdx$. 

By Taylor's theorem, there is some $z_*$ in $[0,z]$ so that 
\begin{equation*}
\Adx'(z) = \Adx'(0) + z \Adx''(z_*) \geq -|\Adx'(0)| + \lambda z.
\end{equation*}
This means that for $z > |\Adx'(0)|/\lambda$, this directional derivative must be positive: in particular, since the geometric median $\gm$ is by definition a point where $\gradAd(\gm) = \vec{0}$, $\gm$ can lie no farther than $|\Adx'(0)|/\lambda$ from the origin. 

\subsection{A probabilistic bound on $\|\gradAdx(\vec{0})\| = \norm{\sum \omega_i \hat{x}_i}$}

We want to bound the norm of the gradient $\gradAdx(\vec{0})$, which we recall from Lemma~\ref{lem:gradad and had} is equal to $\sum \omega_i \hat{x}_i$. We will start with a lemma which helps us understand the effect of variable weights $\omega_i$. 

\begin{lemma}
For any collection of $n$ non-negative real numbers $w_i$, if we define $\omega_i = \nicefrac{w_i}{\sum w_i}$,
\begin{equation*}
n \geq 1 + n^2 \Var(\omega_i) = n \sum \omega_i^2 \geq 1,
\end{equation*}
where $\Var(\omega_i)$ is the variance of $\{\omega_1, \dots, \omega_n\}$. We have equality on the left precisely when all but one of the $w_i$ equal zero and equality on the right precisely when all the $w_i$ are equal. 
\label{lem:mysteryweight}
\end{lemma}

\begin{proof}[Proof of Lemma]
Starting with the definition of variance, and remembering that $\sum \omega_i = 1$, 
\begin{equation*}
\Var(\omega_i) = \frac{1}{n} \sum \omega_i^2 - \left( \frac{1}{n} \sum \omega_i \right)^2 
= \frac{1}{n} \sum \omega_i^2 - \frac{1}{n^2}
\end{equation*}
Solving for $\sum \omega_i^2$,  
\begin{equation*}
\sum \omega_i^2 = \frac{1 + n^2 \Var(\omega_i)}{n}
\end{equation*}
which proves the central equality. Since $\Var(\omega_i) \geq 0$ with equality precisely when all the $\omega_i$ are equal, the inequality on the right follows easily. 

To prove the inequality on the left, we invoke the Bhatia-Davis inequality~\cite{Bhatia:2000ge}, which says that since the $0 \leq \omega_i \leq 1$ and the mean of the $\omega_i$ is $\nicefrac{1}{n}$, we have $\Var(\omega_i) \leq (1 - \frac{1}{n})(\frac{1}{n} - 0)$ with equality precisely when one $\omega_i = 1$ and the remainder are zero. 
\end{proof}

Now we can give our first result:
\begin{proposition}
If we have $n$ points $\hat{x}_i$ sampled independently and uniformly from $S^{d-1}$, and $n$ weights $\omega_i \geq 0$ with $\sum_i \omega_i = 1$ and $\Omega = \max_i \omega_i$, then for any $t > 0$
\begin{equation*}
\mathcal{P}\left(\norm{\sum \omega_i \hat{x}_i} > t \right) \leq d \exp\left(- \frac{3 n t^2}{
 2 n t \Omega \sqrt{d}  + 6 (1 + n^2 \Var \omega_i)} \right).
\end{equation*}
If the $\omega_i$ are all equal (the polygon is equilateral), this simplifies to
\begin{equation*}
\mathcal{P}\left(\norm{\frac{1}{n} \sum \hat{x}_i} > t \right) \leq d \exp\left(- \frac{3 n t^2}{6+2 t \sqrt{d}}\right).
\end{equation*}
\label{prop:ftc}
\end{proposition}

\begin{proof}
We will use Bernstein's inequality~\cite[Theorem~1.2]{Dubhashi:2009ho}: Suppose $X_1, \dots, X_n$ are independent random variables with $X_i - \mathcal{E}(X_i) \leq b$ for each $i$, the variance of each $X_i$ is given by $\sigma_i^2$, and $X = \sum X_i$ (with variance $\sigma^2 = \sum \sigma_i^2$). Then for any $t > 0$,
\[
\mathcal{P}\left( X > \mathcal{E}(X) + t \right) \leq \exp\left( -\frac{t^2}{2 \sigma^2 \left( 1 + \frac{b t}{3 \sigma^2} \right)} \right).
\]
For any unit vector $\vec{v}$, we can set $X_i = \left< \omega_i \hat{x}_i,\vec{v} \right>$. These random variables clearly have expectation 0 and $X_i - \mathcal{E}(X_i) \leq \omega_i \leq \Omega$. Using cylindrical coordinates on $S^{d-1}$ with axis $\vec{v}$, the variance is computed by the integral
\[
\sigma_i^2 = \frac{1}{\sqrt{\pi}}\frac{\Gamma(\frac{d}{2})}{\Gamma(\frac{d - 1}{2})} \int_{-1}^1
 (\omega_i x)^2 (1 - x^2)^{\frac{d - 3}{2}} \,\dx = \frac{\omega_i^2}{d}.
\]
Using Lemma~\ref{lem:mysteryweight}, this implies
\begin{equation*}
\sigma^2 = \frac{1}{d} \sum \omega_i^2 = \frac{1 + n^2 \Var \omega_i}{d n}.
\end{equation*}
This proves that for any $\vec{v}$, 
\begin{equation*}
\mathcal{P}\left( \left<\sum\omega_i \hat{x}_i,\vec{v} \right> > t \right) 
\leq \exp\left( -n d t^2 \frac{3}{6 + 2\, d n t \, \Omega + 6 n^2 \Var \omega_i} \right).
\end{equation*}
Applying this inequality $d$ times for $\vec{v}=\hat{e}_1, \dots, \hat{e}_d$, and using the union bound, we can bound the $L^\infty$ norm of $\sum \omega_i \hat{x}_i$:
\[
\mathcal{P}\left( \norm{\sum \omega_i \hat{x}_i}_\infty = \max_j \left(\sum\omega_i \hat{x}_i\right)_j > t \right) 
\leq  d \exp\left( -n d t^2 \frac{3}{6 + 2\, d n t \, \Omega + 6 n^2 \Var \omega_i} \right).
\]
But we know that for any $\vec{u}\in\R^d$ we have $\frac{1}{\sqrt{d}} \norm{\vec{u}} \leq \norm{\vec{u}}_\infty$, so 
\begin{equation*}
\mathcal{P}\left( \frac{1}{\sqrt{d}} \norm{\sum \omega_i \hat{x}_i} > t \right) \leq 
\mathcal{P}\left( \norm{\sum \omega_i \hat{x}_i}_\infty > t \right) \leq  d \exp\left( -n d t^2 \frac{3}{6 + 2\, d n t \, \Omega + 6 n^2 \Var \omega_i} \right).
\end{equation*}
Replacing $t$ by $\frac{t}{\sqrt{d}}$ yields the statement of the Proposition.
\end{proof}
 
The terms $\Omega$ and $\Var \omega_i$ in the statement of Proposition~\ref{prop:ftc} at first seem mysterious. However, if you read them in light of Lemma~\ref{lem:mysteryweight}, they become clearer. 

At one extreme, if one $\omega_i$ is close to 1 and the remaining $\omega_j$ are small, the sum $\norm{\sum_i \omega_i \hat{x}_i} \sim 1$ regardless of $n$, and $\norm{\sum_i \omega_i \hat{x}_i}$ cannot concentrate on $0$ as $n \rightarrow \infty$. To see this in the statement of the Proposition, observe that in this case $\Omega$ and $\Var \omega_i$ approach their maximum values of $\Omega \sim 1$ and $1 + n^2 \Var \omega_i \sim n$, the $n$'s in numerator and denominator cancel, and the exponent no longer depends on $n$ at all. 

At the other extreme, if the $\omega_i$ are all equal, $\Omega$ and $\Var \omega_i$ are minimized: $\Omega = \nicefrac{1}{n}$ and ${\Var \omega_i = 0}$. In this case, the denominator in the exponent does not depend on $n$ and $\norm{\sum \omega_i \hat{x}_i}$ concentrates on $0$ as fast as possible. We can compare this result to that of Khoi~\cite{Khoi:2005ch}, who showed in a different sense that the equilateral polygons are the ``most flexible'' of all the fixed edgelength polygons. 

In the middle, if the $\omega_i$ are variable, but the number of comparably large $\omega_i$ increases, $\Omega$ and $\Var \omega_i$ act to slow the rate of concentration, but they do not stop it: $\norm{\sum_i \omega_i \hat{x}_i}$ still concentrates on $0$ as $n \rightarrow \infty$.  

\subsection{A probabilistic bound on $\lmin(\HAdx(\vec{0})) = \lmin\left(I - \sum \omega_i \hat{x}_i \hat{x}_i^T\right)$}

We now want to bound the lowest eigenvalue of the Hessian of $\Adx$ at the origin. Again using Lemma~\ref{lem:gradad and had}, we see that $\HAdx(\vec{0}) = I - \sum \omega_i \hat{x}_i \hat{x}_i^T$, where the quantities being summed are outer products of the vectors $\hat{x}_i$. That is, they are the symmetric, positive semidefinite projection matrices which project to the lines spanned by the $\hat{x}_i$. We now show 

\begin{proposition}
If we have $n$ points $\hat{x}_i$ sampled independently and uniformly from $S^{d-1}$, and $n$ weights $0 \leq \omega_i$ with $\sum_i \omega_i = 1$ and $\Omega = \max_i \omega_i$, then for any $t > 0$
\begin{equation*}
\mathcal{P}\left( \lmin\left(I - \sum \wxxt \right) > \frac{d-1}{d} - t \right) \leq d \exp\left(  -\frac{d}{d-1} \cdot \frac{3 d n t^2}{2 n
   t \Omega d +6(1 + n^2 \Var \omega_i)} \right).
\end{equation*}
If the $\omega_i$ are all equal (the polygon is equilateral), this simplifies to 
\begin{equation*}
\mathcal{P}\left( \lmin\left(I - \frac{1}{n}\sum \xxt \right) > \frac{d-1}{d} - t  \right) \leq d \exp\left(  -\frac{d}{d-1} \cdot \frac{3 d n t^2}{2 t d + 6} \right).
\end{equation*}
\label{prop:had zero}
\end{proposition}

\begin{proof}
The statement is similar to the statement of Proposition~\ref{prop:ftc}, so it should not be surprising that this also follows from a  Bernstein inequality, this time for~matrices~\cite[Theorem~1.4]{Tropp:2012fb}: suppose $X_1, \dots, X_n$ are independent random symmetric $d \times d$ matrices, $\mathcal{E}(X_i) = 0$, $\lmax(X_i) \leq b$, the ``matrix variance'' of each $X_i$ is given by $\sigma_i^2 = \mathcal{E}(X_i^2)$, and $X = \sum X_i$ (with ``scalar variance'' $\sigma^2 = \norm{\sum \sigma_i^2}$). Then for any $t > 0$, 
\begin{equation}
\mathcal{P}\left( \lmax(X) \geq t \right) \leq d \exp\left( -\frac{t^2}{2 \sigma^2 \left( 1 + \frac{b t}{3 \sigma^2} \right)} \right)
\label{eq:matrix bernstein}
\end{equation}
We will set $X_i = \omega_i\left(\xxt - \frac{1}{d} I_d\right)$. These are clearly symmetric $d \times d$ matrices. 

We now prove $\mathcal{E}(X_i) = 0$. Since the $\hat{x}_i$ are uniformly sampled on $S^{d-1}$, their distribution is $O(d)$-invariant. This means we can first average $\hat{x}_i$ over any subgroup of $O(d)$ without changing $\mathcal{E}\left(\xxt\right)$. We'll choose the orthotope group of all $2^d$ possible diagonal matrices $D$ with $D_{ii} = \pm 1$. For any vector $\vec{v} \in \R^d$:
\begin{equation*}
\frac{1}{2^d} \left(\sum_{D} (D\vec{v})(D\vec{v})^T \right)_{ij} = 
\frac{1}{2^d} \sum_{D} D_{ii} D_{jj} v_i v_j
\end{equation*}
Now for each of the 4 possible combinations of signs $D_{ii} = \pm 1$ and $D_{jj} = \pm 1$, there are the same number $2^{d-2}$ of elements of the orthotope group with these signs. If $i \neq j$, two products are $+1$ and two are $-1$ and the terms cancel. If $i=j$ all the products are the same. Thus the average matrix $\frac{1}{2^d} \sum_D (D\vec{v}) (D\vec{v})^T$ is a diagonal matrix with entries $v_i^2$.

Since the expectation of the square of a coordinate of a randomly distributed unit vector on $S^{d-1}$ was computed in the proof of Proposition~\ref{prop:ftc} to be $\nicefrac{1}{d}$, we have $\mathcal{E}\left(\xxt\right) = \frac{1}{d} I_d$, proving that $\mathcal{E}(X_i) = 0$.

We now prove that $\lmax(X_i) \leq \Omega \frac{d-1}{d} $. For any matrix $A$, the eigenvalues of $A + kI_d$ are simply $k$ added to the eigenvalues of $A$~(cf.~\cite[Theorem 2.4.8.1]{Horn:2013tf}). So 
\begin{equation*}
\lmax(X_i) = \omega_i \left( \lmax(\xxt) - \frac{1}{d} \right) = \omega_i \frac{d-1}{d} \leq \Omega \frac{d-1}{d}
\end{equation*}
since the largest eigenvalue of a projection matrix like $\xxt$ is 1.

Next, we want to show that $\sigma_i^2 = \mathcal{E}(X_i^2) = \omega_i^2 \frac{d-1}{d^2} I_d$. A direct computation reveals
\begin{equation*}
X_i^2 = \omega_i^2 \left( \left(1 - \frac{2}{d}\right) \xxt + \frac{1}{d^2} I_d \right)
\end{equation*}
and the result follows from our previous computation that $\mathcal{E}\left(\xxt\right) = \frac{1}{d}I_d$. Summing the $\sigma_i^2$ and taking the operator norm, we get 
\begin{equation*}
\sigma^2 = \frac{d-1}{d^2} \sum \omega_i^2.
\end{equation*}
Plugging $b$ and $\sigma$ into~\eqref{eq:matrix bernstein} yields a bound on the probability that $\lmax(X) > t$ or, since ${\lmax(X) = \lmax(\sum \wxxt) - \frac{1}{d}}$, that $\lmax(\sum \wxxt) > \frac{1}{d} + t$. This completes the proof.
\end{proof}

We note that this concentration inequality is better than Proposition~\ref{prop:ftc}: there is an extra factor of $d$ in the numerator which means that the concentration gets faster as $d$ increases. The effect of variable edgelengths is to slow (or stop) the concentration, just as in Proposition~\ref{prop:ftc}; the same comments on the role of $\Omega$ and $\Var \omega_i$ apply here. 

\subsection{A bound on the change in the radial second derivative}

For any point $\vec{s} \in \R^{d}$, the second derivative of $\Adx$ along the ray through $\vec{s}$ is given by evaluating the Hessian as a quadratic form on the vector $\vec{s}$ itself. Our last proposition gave us a lower bound on the result at the origin; we now show that this can't change too fast as we move away from the origin. 

\begin{proposition}
For $\norm{\vec{s}} < 1$ we have 
\begin{equation*}
\frac{\left< \HAdx(\vec{s}) \vec{s}, \vec{s} \right>}{\left<\vec{s},\vec{s}\right>} - \frac{\left< \HAdx(0) \vec{s},\vec{s} \right>}{\left< \vec{s}, \vec{s}\right>} \geq -\norm{\vec{s}} \frac{6 + \norm{\vec{s}} + \norm{\vec{s}}^2}{(1 - \norm{\vec{s}})^3}.
\end{equation*}
Since the fraction at right is increasing in $\norm{\vec{s}}$, we can easily simplify the statement given a better upper bound on $\norm{\vec{s}}$. In particular, for $\norm{\vec{s}} < \nicefrac{1}{50}$, the right-hand side $\geq - 7 \norm{\vec{s}}$.
\label{prop:change in hessian}
\end{proposition}

\begin{proof}
Using Lemma~\ref{lem:gradad and had}, we see that
\begin{multline*}
\left< \HAdx(\vec{s}) \vec{s}, \vec{s} \right> - \left< \HAdx(0) \vec{s},\vec{s} \right> = \\ 
\left(\!\left( \sum \frac{\omega_i}{\norm{\hat{x}_i - \vec{s}}}\right) - 1\right) \left< \vec{s}, \vec{s} \right> - 
\sum \omega_i \left( \frac{\left<\hat{x}_i - \vec{s},\vec{s}\right>^2}{\norm{\hat{x}_i-\vec{s}}^3} - \left< \hat{x}_i, \vec{s}\right>^2 \right).
\end{multline*}
Using the estimates $1 - \norm{\vec{s}} \leq \norm{\hat{x}_i - \vec{s}} \leq 1 + \norm{\vec{s}}$ and recalling that $\sum \omega_i = 1$, we can underestimate the right hand side by 
\begin{equation*}
-\frac{\norm{\vec{s}}^3}{1 + \norm{\vec{s}}} - \sum \omega_i \left( \frac{\left<\hat{x}_i - \vec{s},\vec{s}\right>^2}{(1 - \norm{\vec{s}})^3} - \left< \hat{x}_i, \vec{s}\right>^2 \right) \geq 
-\frac{\norm{\vec{s}}^3}{1 + \norm{\vec{s}}} - \frac{5 \norm{\vec{s}}^3 + \norm{\vec{s}}^4 + \norm{\vec{s}}^5}{(1 - \norm{\vec{s}})^3}
\end{equation*}
where the second part follows from finding a common denominator, expanding, and cancelling, using Cauchy--Schwartz carefully to underestimate the inner product terms as needed. Observing that $1 + \norm{\vec{s}} > 1 > (1 - \norm{\vec{s}})^3$ allows us to underestimate $-\nicefrac{\norm{\vec{s}}^3}
{1 + \norm{\vec{s}}} \geq -\nicefrac{\norm{\vec{s}}^3}{(1 - \norm{\vec{s}})^3}$, completing the proof. \end{proof}
%Nb: This argument is in geometric-median-concentration-attempt-8.nb

\subsection{Bounding the norm of the geometric median}

We are now in a position to bound the norm of the geometric median! This will proceed in two stages: first, we'll use the Poincar\'e--Hopf index theorem to show that $\norm{\gm} < \nicefrac{1}{50}$ under certain hypotheses. Then we can immediately bootstrap to get a sharper bound.

\begin{proposition} 
If $\norm{\sum \omega_i \hat{x}_i} = \|\gradAdx(\vec{0})\| < \nicefrac{5}{1000}$, $\lmin(\HAdx(\vec{0})) > \frac{d-1}{d} - \frac{1}{100}$, and $d \geq 2$, then $\norm{\gm} \leq \nicefrac{1}{50}$. 
\label{prop:preparatory bound}
\end{proposition}

\begin{proof}
Given our hypothesis on $\lmin$ of the Hessian, we know that the $\hat{x}_i$ are not all colinear. This means that $\gm$ is the unique point inside $S^{d-1}$ where the vector field $\gradAdx$ vanishes. We will now show that $\gradAdxy$ has a zero inside the sphere of radius $\nicefrac{1}{50}$; by uniqueness, this point must be the geometric median. 

Along any ray from the origin, we may restrict $\Adx$ to a scalar function $\Adx(z)$. Using Proposition~\ref{prop:change in hessian}, on the interval $[0,\nicefrac{1}{50}]$ our hypotheses imply
\begin{equation*}
\Adx'(0) \geq -\frac{5}{1000} \quad \text{and} \quad \Adx''(z) \geq \frac{1}{2} - \frac{1}{100} - \frac{7}{50} = \frac{7}{20}.
\end{equation*}
By Taylor's theorem, there is some $z_* \in [0,\nicefrac{1}{50}]$ so that
\begin{equation*}
\Adx'(\nicefrac{1}{50}) = \Adx'(0) + \frac{1}{50} \Adx''(z_*) \geq -\frac{5}{1000} + \frac{1}{50}\cdot \frac{7}{20} = \frac{2}{1000} > 0.
\end{equation*}
This means that the directional derivative of $\Adx$ in the outward direction is positive on the boundary of the sphere of radius $\nicefrac{1}{50}$, or that $\gradAdxy$ points outward on this sphere. In particular, this implies that the vector field has index 1 on the sphere, and so by the Poincar\'e--Hopf index theorem must vanish at some point inside the sphere.  
\end{proof}

We can now prove our main theorem. 

\begin{theorem}
If we have $n$ points $\hat{x}_i$ sampled uniformly on $S^{d-1}$ ($d \geq 2$), $n$ weights $\omega_i > 0$ so that $\sum \omega_i = 1$, and $\max \omega_i = \Omega$, then for any $t < \nicefrac{5}{1000}$ we have
\begin{equation}
\mathcal{P}\left( \norm{\gm} < \frac{t}{\frac{d-1}{d} - \frac{3}{20}} \right) \geq
1 - 2 d \exp\left(- \frac{3 n t^2}{
 2 n t \Omega \sqrt{d}  + 6 (1 + n^2 \Var \omega_i)} \right). %+ d \exp\left(  -\frac{d}{d-1} \cdot \frac{3 d n t^2}{2 n
%   t \Omega d +6(1 + n^2 \Var \omega_i)} \right).
\label{eq:main bound}
\end{equation}
If all the $\omega_i$ are equal (the polygon is equilateral)
\begin{equation*}
\mathcal{P}\left( \norm{\gm} < \frac{t}{\frac{d-1}{d} - \frac{3}{20}} \right) \geq 1 - 2 d \exp\left( - \frac{3 n t^2}{2 \sqrt{d} t + 6} \right).
\end{equation*}
For $d = 3$, we have the further simplification
\begin{equation*}
\mathcal{P}\left( \norm{\gm} < t \right) \geq 1 - 6 \exp\left( -\nicefrac{n t^2}{9} \right).
\end{equation*}
\label{thm:main}
\end{theorem}

\begin{proof}
We first define two random events: $\lmin(\HAdx(\vec{0})) > \frac{d-1}{d} - \frac{1}{100}$ (event $A$) and $\norm{\gradAdx(\vec{0})} < t < \nicefrac{5}{1000}$ (event $B$), which will happen for some choices of $\hat{x}_i$. Suppose both events occur. 

As in Proposition~\ref{prop:preparatory bound}, we restrict $\Adx$ to a scalar function $\Adx(z)$ on a ray; this time, the ray is assumed to pass through $\gm$. By Taylor's theorem, if we evaluate at $z = \norm{\gm}$, there is some $0\leq z_* \leq \norm{\gm}$ so that  
\begin{equation}
0 = \Adx'(\norm{\gm}) = \Adx'(0) + \norm{\gm} \Adx''(z_*). 
\label{eq:taylor theorem setup}
\end{equation}
Since we are assuming $A \land B$, the hypotheses of~Proposition~\ref{prop:preparatory bound} are satisfied and $\norm{\gm} < \nicefrac{1}{50}$. In turn, this means that  Proposition~\ref{prop:change in hessian} holds at $z_*$, and 
\begin{equation*}
\Adx''(z_*) \geq \Adx''(0) - \nicefrac{7}{50} \geq \lmin(\HAdx(\vec{0})) - \nicefrac{7}{50}.
\end{equation*}
Since $A$, we have $\Adx''(z_*) > \frac{d-1}{d} - \frac{3}{20}$. 
As before, since $\Adx'(\vec{0})$ is a directional derivative, it satisfies $\Adx'(0) \geq -\|\gradAdx(\vec{0})\| > -t$.  We can plug both estimates into~\eqref{eq:taylor theorem setup} and solve for $\norm{\gm}$, obtaining
\begin{equation*}
\norm{\gm} < \frac{t}{\frac{d-1}{d} - \frac{3}{20}}.
\end{equation*} 
If we call this event $C$, we have shown that $A \land B \implies C$, and hence that $\mathcal{P}(C) \geq \mathcal{P}(A \land B)$. This means that 
\begin{equation}
\mathcal{P}(\lnot C) \leq \mathcal{P}(\lnot(A \land B)) = \mathcal{P}(\lnot A \lor \lnot B) \leq \mathcal{P}(\lnot A) + \mathcal{P}(\lnot B).
\label{eq:logic}
\end{equation}
Now $\mathcal{P}(\lnot A)$ was bounded above in Proposition~\ref{prop:had zero}, while $\mathcal{P}(\lnot B)$ was bounded above in Proposition~\ref{prop:ftc}. We now compare these upper bounds, noting that we have chosen $t_* = \frac{1}{100}$ in the statement of Proposition~\ref{prop:had zero} while the $t$ in Proposition~\ref{prop:ftc} is smaller -- less than $\frac{5}{1000} = \frac{1}{200}$. The bounds are
\begin{equation*}
d \exp\left(  -\frac{d}{d-1} \cdot \frac{3 d n t_*^2}{2 n
   t_* \Omega d +6(1 + n^2 \Var \omega_i)} \right) \quad\text{and}\quad 
 d \exp\left(- \frac{3 n t^2}{
 2 n t \Omega \sqrt{d}  + 6 (1 + n^2 \Var \omega_i)} \right).
\end{equation*}
Of course, it suffices to compare the absolute values of the fractions inside the exponential functions (since both are negative). We can simplify the comparison by rewriting these as
\begin{equation*}
 \frac{d}{d-1} \cdot \frac{3 n t_*}{2 n
    \Omega  + \frac{6(1 + n^2 \Var \omega_i)}{d t_*}}
 \quad\text{and}\quad
 \frac{3 n t}{
 2 n \Omega \sqrt{d}  + \frac{6 (1 + n^2 \Var \omega_i)}{t}}.
\end{equation*}
It is now evident that if we compare the right fraction with the second fraction on the left, the numerator on the right is smaller and each term in the denominator is larger (recall $t < t_*$). Multiplying by $\frac{d}{d-1} > 1$ makes the left hand side even larger. Restoring the minus sign reverses this conclusion, and we see that our bound on $\mathcal{P}(\lnot B)$ is larger than our bound on $\mathcal{P}(\lnot A)$, as claimed. Returning this conclusion to~\eqref{eq:logic}, we see $\mathcal{P}(\lnot C) \leq 2 \, \mathcal{P}(\lnot B)$, which is the first statement of the Theorem.

The simplification when all the $\omega_i = \nicefrac{1}{n}$ is an immediate consequence. To simplify to $d=3$, we observe that $\frac{t}{\nicefrac{2}{3} - \nicefrac{3}{20}} = \frac{60}{31} t$; substituting $t \rightarrow \frac{31}{60} t$ on the right hand side yields an expression in the form $1 - 6 \exp(-f(t) n t^2)$, where $f(t)$ is a rational function bounded below by $\nicefrac{1}{9}$ for $t \in [0,\nicefrac{5}{1000}]$. 
\end{proof}

We now make a few remarks. First, if you carefully examine Proposition~\ref{prop:change in hessian}, the lower bound on $\Adx''(z)$ improves as $z \rightarrow 0$. One can wring some extra information out of this, but the improvement in the final bound is minimal. Similarly, it is clear that one could set $t_* < \frac{1}{100}$ in our bound on $\mathcal{P}(\lnot A)$ without losing the conclusion, as long as $t_* < t$. Again, this does not significantly improve things. 

\section{Distances and Angles}

We now want to restate our main Theorem~\ref{thm:main} in terms of the chordal and max-angular distance from a random arm to the nearest closed polygon using Proposition~\ref{prop:distance bound}.

\begin{corollary}
If we have $n$ points $\hat{x}_i$ sampled uniformly on $S^{d-1}$ ($d \geq 2$), $n$ weights $\omega_i > 0$ so that $\sum \omega_i = 1$, and $\max \omega_i = \Omega$, then for any $t < \nicefrac{5}{1000} \cdot \nicefrac{1}{\sqrt{2\sum\omega_i^2}}$ we have
\begin{equation*}
\mathcal{P}\left(d_\text{chordal}(\pmb{x},\Pol(n,d,w)) < \frac{t}{\frac{d-1}{d} - \frac{3}{20}}\right) \geq 1 - 2 d \exp\left( \frac{-3 t^2}{3 + t \Omega \sqrt{\frac{2 d n}{1 + n^2 \Var \omega_i}}} \right).
\end{equation*}
If all the $\omega_i$ are equal (the polygon is equilateral), for $t < \nicefrac{5}{1000} \cdot \sqrt{\nicefrac{n}{2}}$ we have 
\begin{equation*}
\mathcal{P}\left(d_\text{chordal}(\pmb{x},\Pol(n,d,1)) < \frac{t}{\frac{d-1}{d} - \frac{3}{20}}\right) \geq 1 - 2 d \exp\left( \frac{-t^2}{1 + \frac{\sqrt{d}}{600}} \right).
\end{equation*}
In dimension 3, this simplifies (again, for $t < \nicefrac{5}{1000} \cdot \sqrt{\nicefrac{n}{2}}$), as 
\begin{equation*}
\mathcal{P}\left(d_\text{chordal}(\pmb{x},\Pol(n,3,1)) < t \right) \geq 1 - 6 \exp\left( \nicefrac{-t^2}{4} \right).
\end{equation*}
\label{cor:chordal concentration}
\end{corollary}

\begin{figure}[t!]
	\centering
		\includegraphics[width=2.5in]{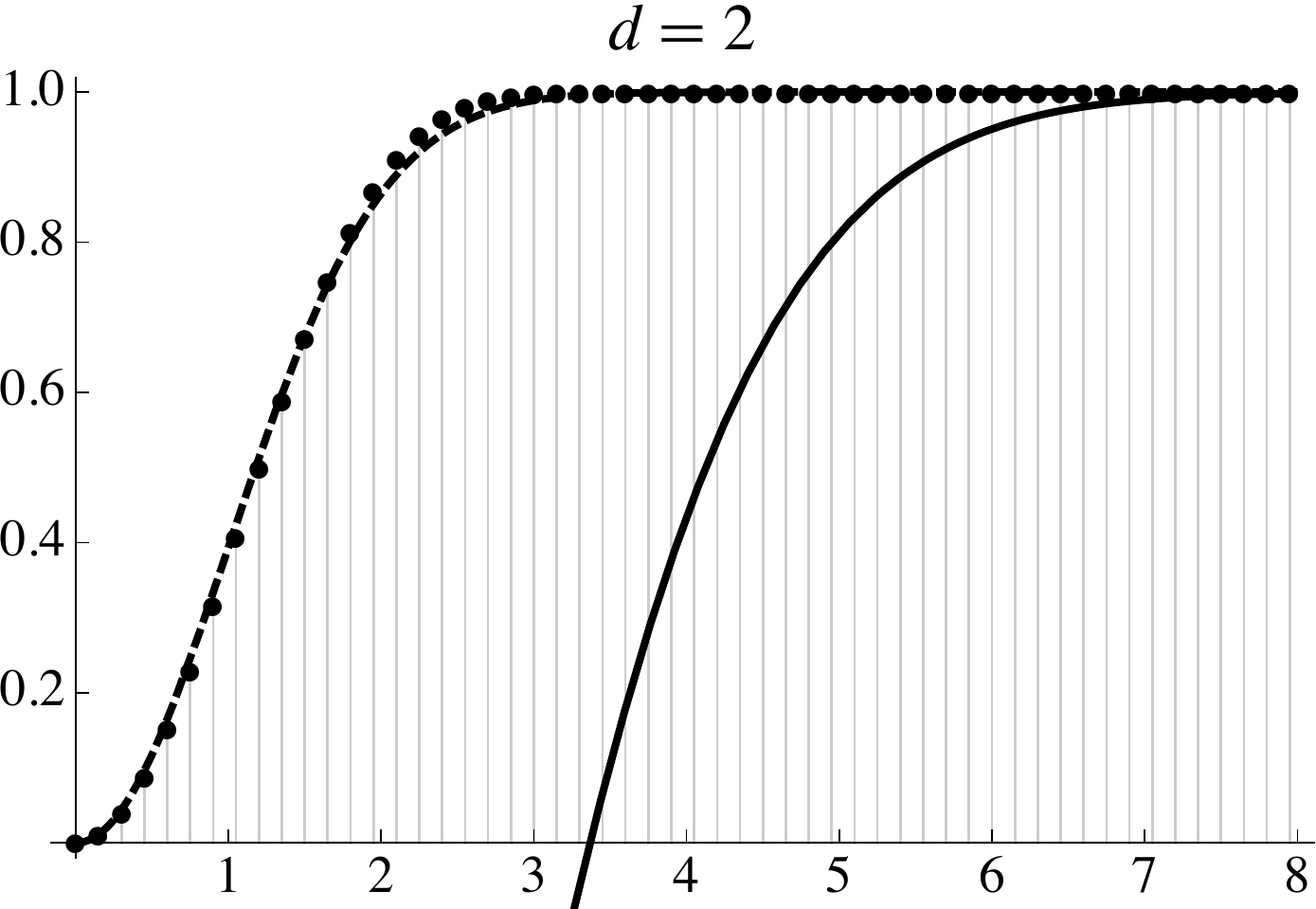}\hspace{.4in}
		\includegraphics[width=2.5in]{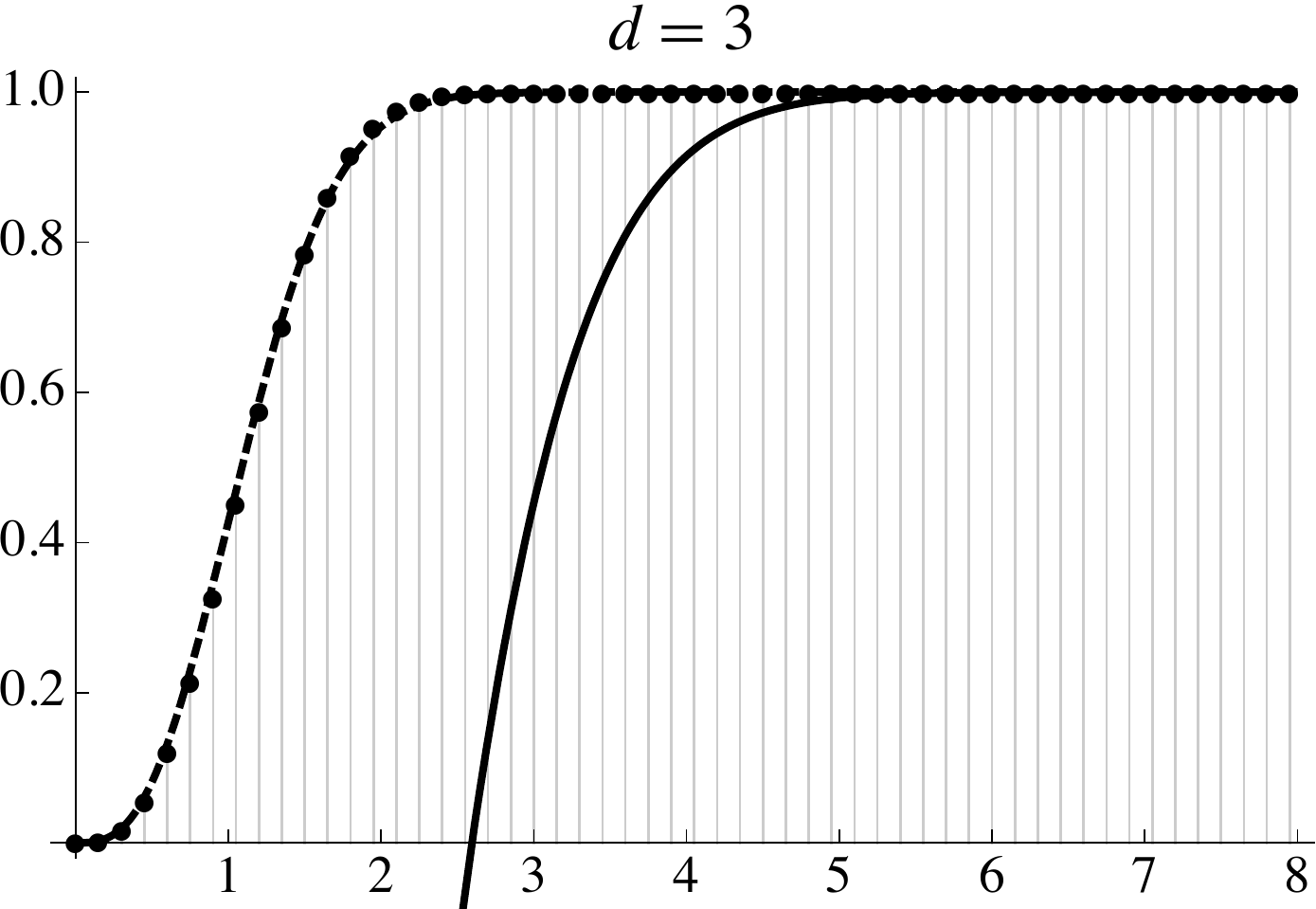}
		                            
		\vspace{.2in}               
		\includegraphics[width=2.5in]{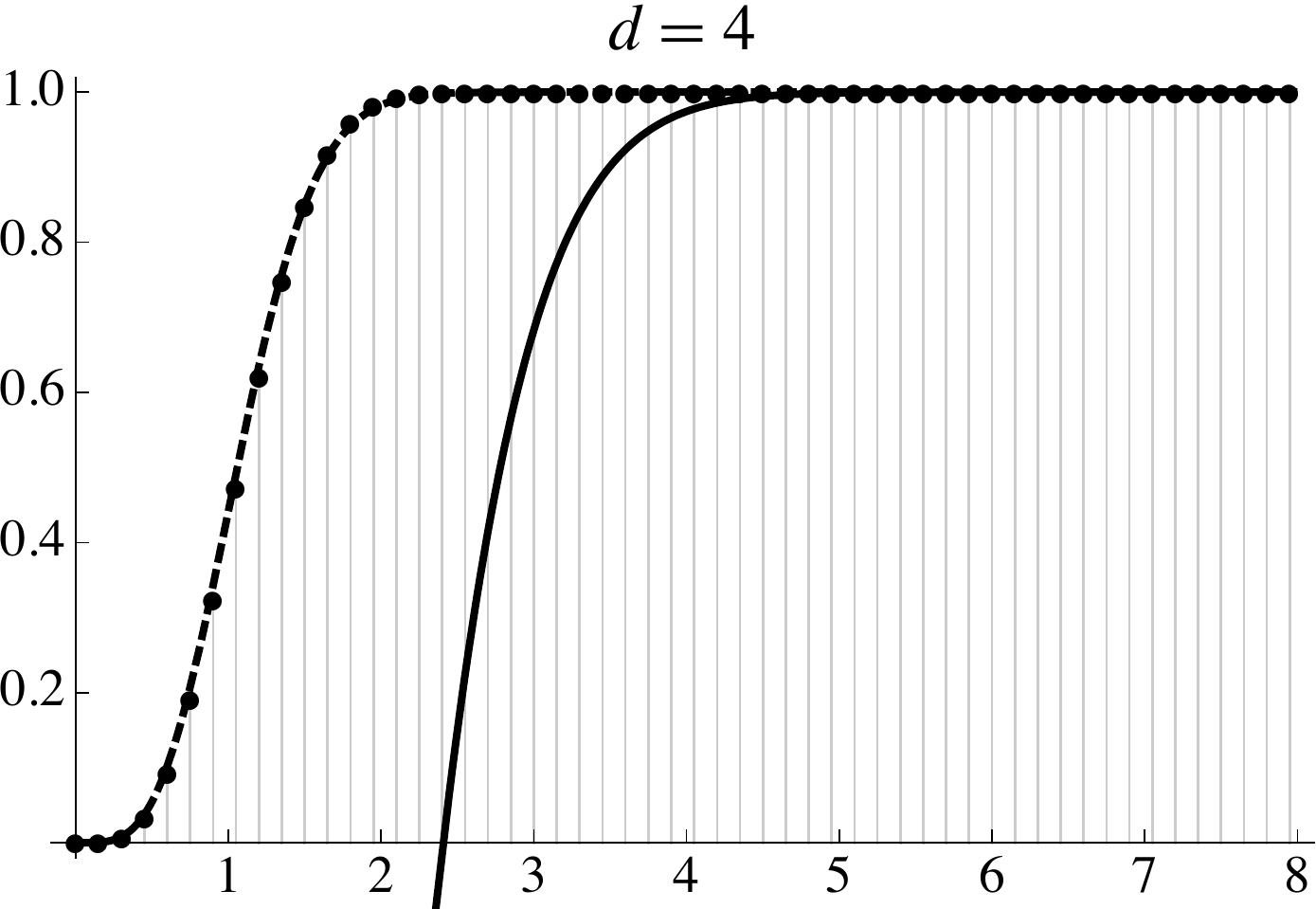}\hspace{.4in}
		\includegraphics[width=2.5in]{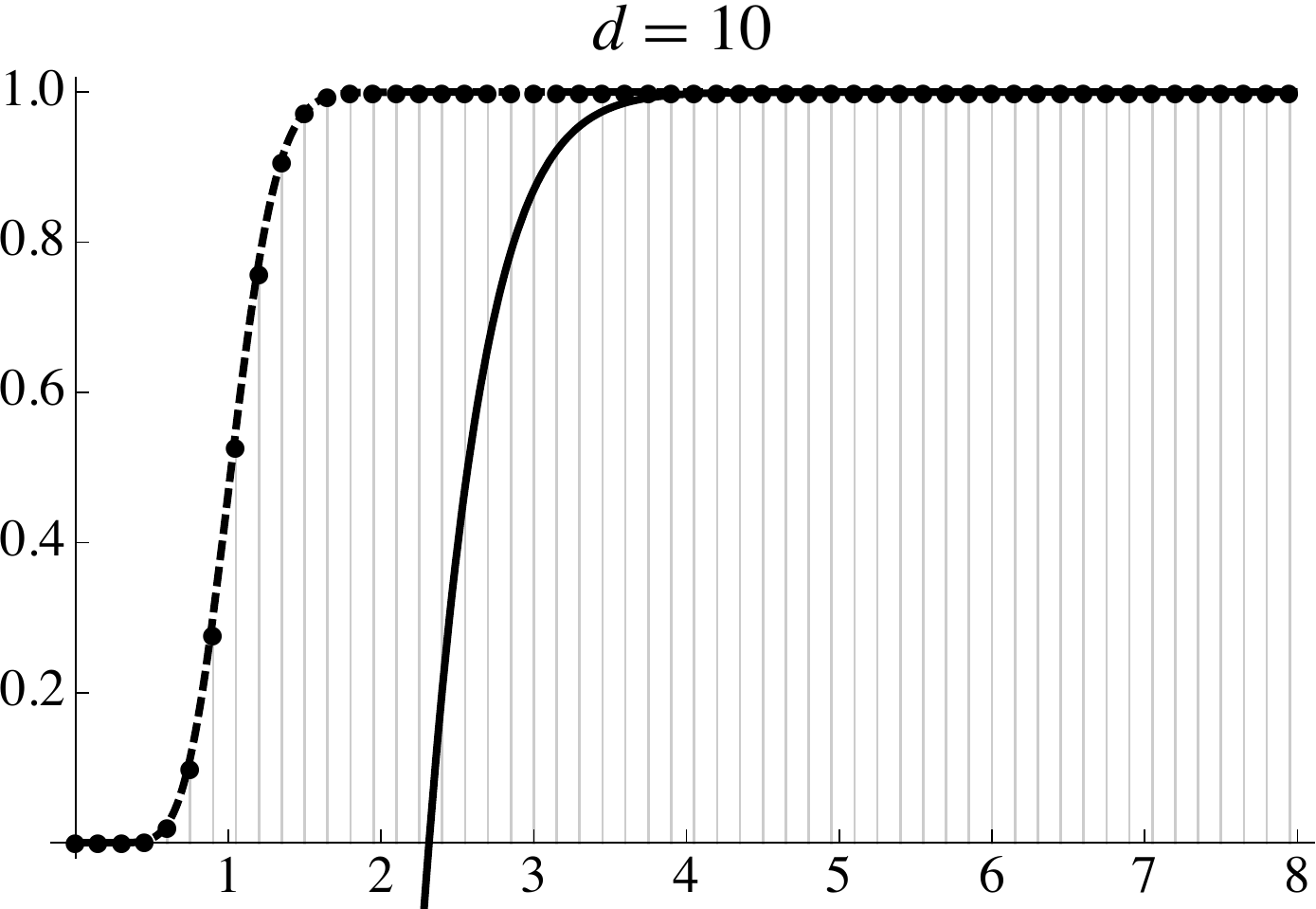}
	\caption{For $d=2,3,4,10$, we generated 250,000 random elements of $\Arm(10,d,1)$. The plots show the implied bound from Corollary~\ref{cor:chordal concentration} (solid), the empirical CDF of chordal distance to closure for those samples which were median-closeable (dots), and the CDF of the Nakagami$\left(\frac{d}{2},\frac{d}{d-1}\right)$ distribution (dashed) given by Proposition~\ref{prop:dchordal asymptotics} for the large-$n$ limit (which is only slightly different, even though $n=10$ is quite small). Though the hypotheses of Corollary~\ref{cor:chordal concentration} are only satisfied when $t<\frac{5}{1000}\sqrt{5}\approx 0.01118$, the data strongly suggests that the bound is valid on a much larger range. We see from the plots that the bound cannot be dramatically improved.}
	\label{fig:bound vs data}
\end{figure}

The problem with Corollary~\ref{cor:chordal concentration} is that the hypotheses (on $t$) are disappointingly restrictive: for $\Arm(n,3,1)$, we need $n > 538,519$ to extend the domain of $t$ to the point where the right-hand side becomes positive! On the other hand, numerical experiments (Figure~\ref{fig:bound vs data}) comparing our bounds to experimental data and to the large-$n$ Nakagami distribution proved in Proposition~\ref{prop:dchordal asymptotics} show that the conclusions of Corollary~\ref{cor:chordal concentration} cannot be made much stronger.  Further, these experiments suggest that, at least in the equilateral case, one should be able to entirely remove the upper bound on $t$-- we leave this as 
\begin{conjecture}
The conclusions of Corollary~\ref{cor:chordal concentration} hold for any $t > 0$. 
\end{conjecture}

We now proceed to prove Corollary~\ref{cor:chordal concentration}. 

\begin{proof}[Proof of Corollary~\ref{cor:chordal concentration}]
Proposition~\ref{prop:distance bound} tells us $d_\text{chordal}(\pmb{x},\Pol(n,d,w))<\sqrt{2 \sum \omega_i^2} \norm{\gm}$, so to get a bound on the probability that $d_\text{chordal}(\pmb{x},\Pol(n,d,w)) < \frac{t}{\frac{d-1}{d} - \frac{3}{20}}$ we need to 
make the substitution $t \rightarrow t \sqrt{2 \sum \omega_i^2}$ on the right hand side of~\eqref{eq:main bound}. Recalling that Lemma~\ref{lem:mysteryweight} shows $\sum \omega_i^2 = \frac{1 + n^2 \Var \omega_i}{n}$ and carefully simplifying yields the first result. 

For the second result, it follows immediately from the assumption that $\omega_i = \frac{1}{n}$ that the first result simplifies to
\begin{equation*}
\mathcal{P}\left(d_\text{chordal}(\pmb{x},\Pol(n,d,1)) < \frac{t}{\frac{d-1}{d} - \frac{3}{20}}\right) \geq 1 - 2 d \exp\left( \frac{-3 t^2}{3 + t\sqrt{\frac{2 d}{n}}} \right).
\end{equation*}
Using our upper bound on $t$, we see that the right hand side obeys
\begin{equation*}
1 - 2 d \exp\left( \frac{-3 t^2}{3 + t\sqrt{\frac{2 d}{n}}} \right) > 
1 - 2 d \exp\left( \frac{-3 t^2}{3 + \frac{\sqrt{d}}{200}} \right) 
\end{equation*}
%Note: checking that the inequalities are in the right sense is done, painfully, in geometric-median-concentration-attempt-8.nb, near the end.
which immediately implies the second result. 

For the third result, we simplify the fraction on the left hand side and substitute $t \rightarrow \frac{31}{60} t$ as we did above in the simplification of Theorem~\ref{thm:main}; the complicated constant that results as the coefficient of $t^2$ in the exponent is slightly less that $-\nicefrac{1}{4}$. 
\end{proof}

The statements for the maximum angular change in edge direction are similar, but somewhat easier to prove because the relationship between $\norm{\gm}$ and the max-angular distance is simpler. 

\begin{corollary}
If we have $n$ points $\hat{x}_i$ sampled uniformly on $S^{d-1}$ ($d \geq 2$), $n$ weights $\omega_i > 0$ so that $\sum \omega_i = 1$, and $\max \omega_i = \Omega$, then for any $t < \nicefrac{5}{1000}$ we have
\begin{equation*}
\mathcal{P}\left(d_\text{max-angular}(\pmb{x},\Pol(n,d,w)) < \frac{t}{\frac{d-1}{d} - \frac{3}{20}}\right) \geq 1 - 2 d \exp\left(- \frac{13 n t^2}{
 9 n t \Omega \sqrt{d}  + 30 (1 + n^2 \Var \omega_i)} \right). 
\end{equation*}
If all the $\omega_i$ are equal (the polygon is equilateral), for $t < \nicefrac{5}{1000}$ we have 
\begin{equation*}
\mathcal{P}\left(d_\text{max-angular}(\pmb{x},\Pol(n,d,1)) < \frac{t}{\frac{d-1}{d} - \frac{3}{20}}\right) \geq 1 - 2 d \exp\left( \frac{-26 n t^2}{60 + \frac{9\sqrt{d}}{100}} \right).
\end{equation*}
In dimension 3, this simplifies (again, for $t < \nicefrac{5}{1000}$), as 
\begin{equation*}
\mathcal{P}\left(d_\text{max-angular}(\pmb{x},\Pol(n,3,1)) < t \right) \geq 1 - 6 \exp\left( -\nicefrac{n t^2}{9} \right).
\end{equation*}
\label{cor:angular concentration}
\end{corollary}

\begin{proof}
We know from Proposition~\ref{prop:distance bound} that $d_\text{max-angular}(\pmb{x},\Pol(n,d,w)) < \arcsin \norm{\gm}$. Since we're only going to apply this bound when $\norm{\gm} < \frac{t}{\frac{d-1}{d} - \frac{3}{20}} < \frac{1}{70}$ (since $d \geq 2$ and $t \leq \nicefrac{5}{1000}$), we can safely make the overestimate $\arcsin \norm{\gm} \leq \frac{14}{13} \norm{\gm}$. 

Substituting $t \rightarrow \frac{13}{14}t$ in~\eqref{eq:main bound} leads us to replace $3$ by $3 (\frac{13}{14})^2 < 2.6$ in the coefficient of $n t^2$ in the numerator and $2$ by $2 (\nicefrac{13}{14}) > 1.8$ in the coefficient of $t$ in the denominator. Simplifying gives us the first statement. 

To reach the second statement, we first observe that $\omega_i = \frac{1}{n}$ means $\Var \omega_i = 0$ and $\Omega = \frac{1}{n}$. Substituting these into the first statement (and overestimating the $t$ in the denominator by $\nicefrac{5}{1000}$) yields the result. 

Finally, the third statement (as in the proof of Corollary~\ref{cor:chordal concentration}) requires us to substitute $t \rightarrow \frac{31}{60} t$ to simplify the left-hand side. The resulting complicated coefficient of $n t^2$ on the right-hand side is about $-0.115521 < -\nicefrac{1}{9}$.
\end{proof}

\section{Discussion}

\begin{figure}[t!]
	\centering
		\includegraphics[width=6.1in]{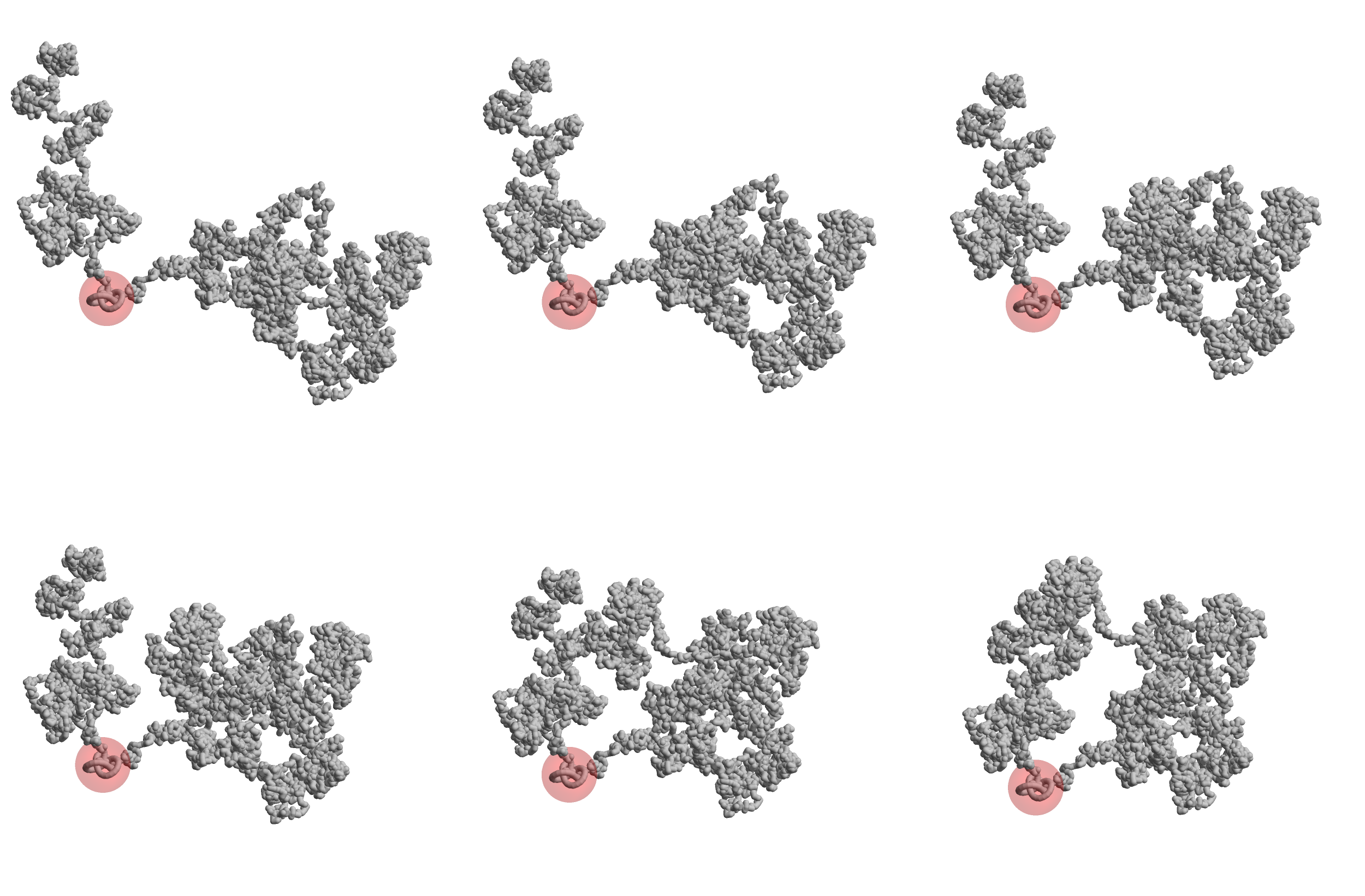}
	\caption{A 10,000 step equilateral arm in $\mathbb{R}^3$ containing a small trefoil (top left) and its geometric median closure (bottom right). The intermediate images show equally-spaced points along the geodesic between the arm and its closure in $\Arm(10,\!000,3,1)$. The failure to close of the arm is $\approx 101.118$ and the geometric median has norm $\|\gm\|\approx 0.0151318$. The chordal distance between the arm and its closure is $\approx 1.23696$ and $d_\text{max-angular}\approx 0.0151324$, which agrees with the bound $\arcsin \|\gm\|$ to eleven decimal places.}
	\label{fig:trefoil closure}
\end{figure}

From Corollary~\ref{cor:angular concentration} we see that closing a random arm is unlikely to change any edge very much. In particular, we should expect local features to be preserved by closure, as in the case of the local trefoil knot shown in \figr{trefoil closure}. This suggests that closing up an arm is unlikely to destroy any local knots: in other words, the probability of local knotting in the standard measure on $\Arm(n,3,w)$ should be essentially the same as the probability of local knotting in the pushforward measure on $\Pol(n,3,w)$ via the map $\pmb{x} \mapsto \gmc(\pmb{x})$.

Of course, this map is not defined on all of $\Arm(n,d,w)$, but we know from \thm{main} that it is defined on all but an exponentially small fraction of $\Arm(n,d,w)=\prod S^{d-1}(w_i)$; pushing forward the restriction of the product measure to the domain of $\gmc$ produces what we'll call the \emph{pushforward measure} on $\Pol(n,d,w)$. On the other hand, the standard probability measure on $\Pol(n,d,w)$ is simply the volume measure induced by the Riemannian metric it inherits from $\Arm(n,d,w)$. Since we've seen in Corollaries~\ref{cor:chordal concentration} and~\ref{cor:angular concentration} that almost all of $\Arm(n,d,w)$ is within a fixed distance of $\Pol(n,d,w)$, it is reasonable to expect that this pushforward measure is close to the standard measure.

Indeed, this seems to be true. Rayleigh~\cite{Rayleigh:1919do} showed that the distribution of end-to-end distances in a random element of $\Arm(n,3,1)$ is
\begin{equation*}
\Phi_n(\ell) = \frac{1}{2 \pi^2 \ell} \int_0^\infty x \sin \ell x \operatorname{sinc}^n x\, dx.
\end{equation*}
We note that a closed form for $\Phi_n$ is classical (see~\cite[2.181]{hughes1995random}). Since a random closed polygon is formed from two random arms, conditioned on the hypothesis that their end-to-end distances are the same, the pdf of the length of the chord connecting vertices $0$ and $k$ in an polygon of $n$ edges turns out to be given by 
\begin{equation*}
\operatorname{Chord}_{n,k}(\ell) = \frac{1}{C(n)} 4 \pi \ell^2 \Phi_k(\ell) \Phi_{n-k}(\ell).
\end{equation*}
where the factor of $4 \pi \ell^2$ comes from the fact that vertex $k$ lies on a sphere of radius $\ell$ and $C(n)$ is the volume of polygon space (which is known; see~\cite{Cantarella:2013wl} for an identification between polygon space and a certain polytope which yields an explicit, though complicated, formula for $C(n)$).

%at least in the cases where we know how to characterize the standard measure on polygon space. In~\cite{Cantarella:2013wl} we gave such a characterization for $\Pol(n,3,w)$ in terms of so-called \emph{action-angle coordinates}. Explaining this characterization precisely would require a bit of a digression, but the following is a consequence:
%
%\begin{proposition}\label{prop:chordlengths}
%	If $d_i$ is the distance from the first vertex of an $n$-gon in $\Pol(n,3,w)$ to the $i$th vertex for $i=3,\ldots, n-1$, then the $d_i$ satisfy the inequalities
%	\begin{equation}\label{eq:moment polytope}
%		|w_1 - w_2| \leq d_3 \leq w_1 + w_2 
%		\qquad 
%		\begin{matrix} 
%		w_{i} \leq d_i + d_{i+1} \\
%		|d_i - d_{i+1}| \leq w_{i} 
%		\end{matrix}
%		\qquad
%		|w_n - w_{n-1}| \leq d_{n-1} \leq w_n + w_{n-1}
%	\end{equation}
%	and the vector $(d_3,\ldots , d_{n-1}) \in \R^{n-3}$ is uniformly distributed on the polytope determined by these inequalities.
%\end{proposition}

\begin{figure}[t!]
	\centering
		\includegraphics[width=2.5in]{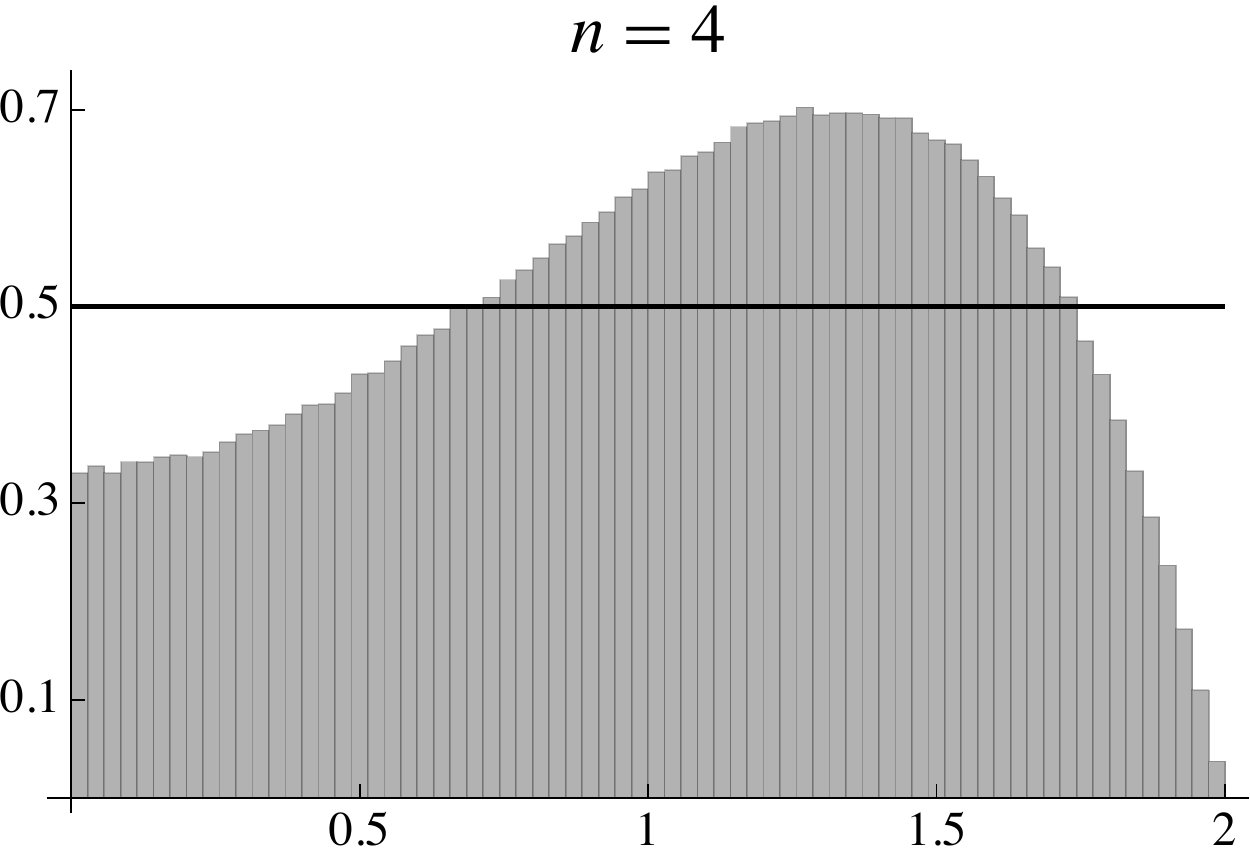}\qquad\qquad\qquad
		\includegraphics[width=2.5in]{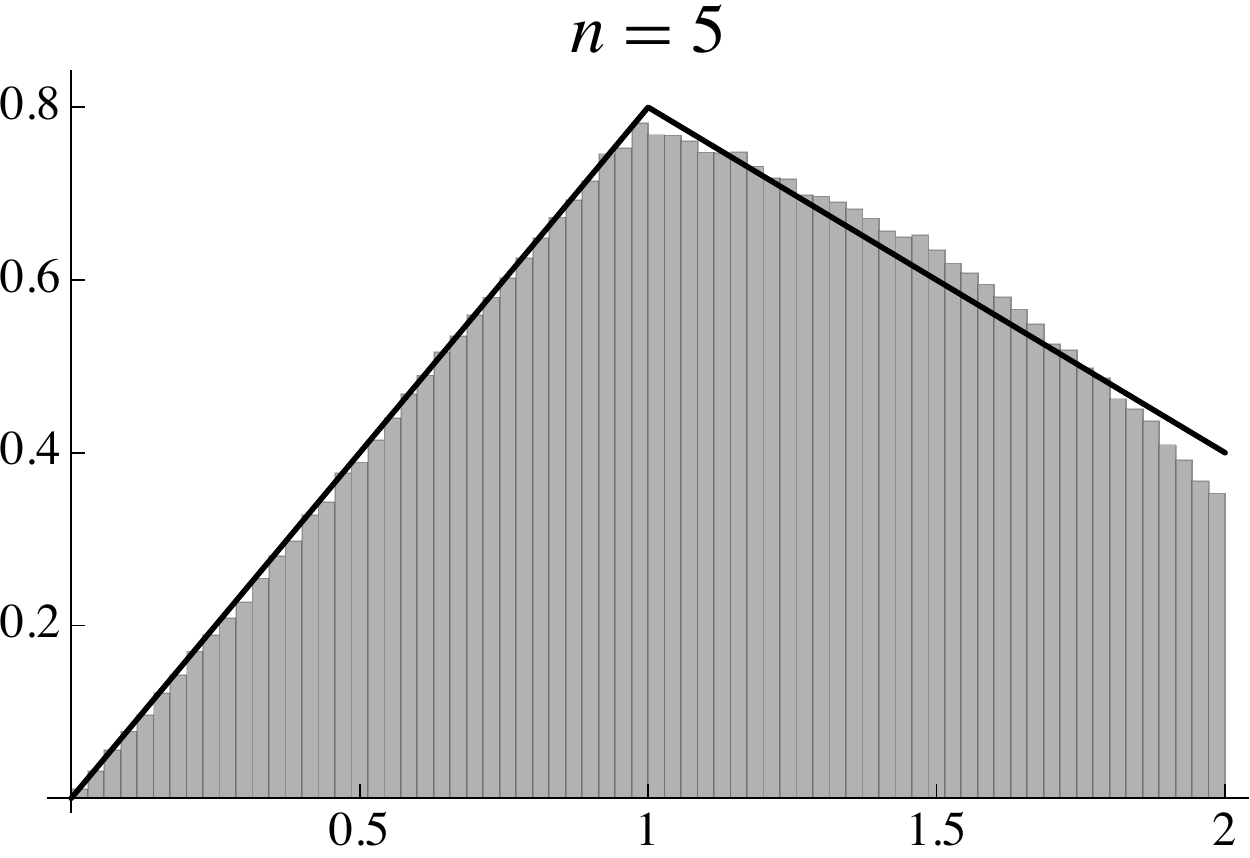}
	\caption{For $n=4,5$, we generated 1,000,000 random equilateral $n$-edge arms in $\R^3$, computed their geometric median closures (when they existed), and then computed the distance from the first to the third vertex in the resulting closed $n$-gon. The histograms show the resulting distributions of chordlengths as well as the density of the chordlength for the standard distribution on $\Pol(n,3,1)$. Closure failed for 2474 quadrilaterals and for 117 pentagons.}
	\label{fig:quadrilateral and pentagon chordlengths}
\end{figure}

Therefore, the extent to which the distributions of the chordlengths match $\operatorname{Chord}_{n,k}$ gives a sense of how close a given distribution on $\Pol(n,3,w)$ is to the standard one. For $n=4$ and $5$, we can see in \figr{quadrilateral and pentagon chordlengths} that the pushforward measure from $\Arm(n,3,1)$ is not particularly close to the standard measure. However, as $n$ increases these statistics cannot distinguish between the pushforward measure and the standard measure; see \figr{10-gon chordlengths}.

\begin{figure}[t!]
	\centering
		\includegraphics[width=1.4in]{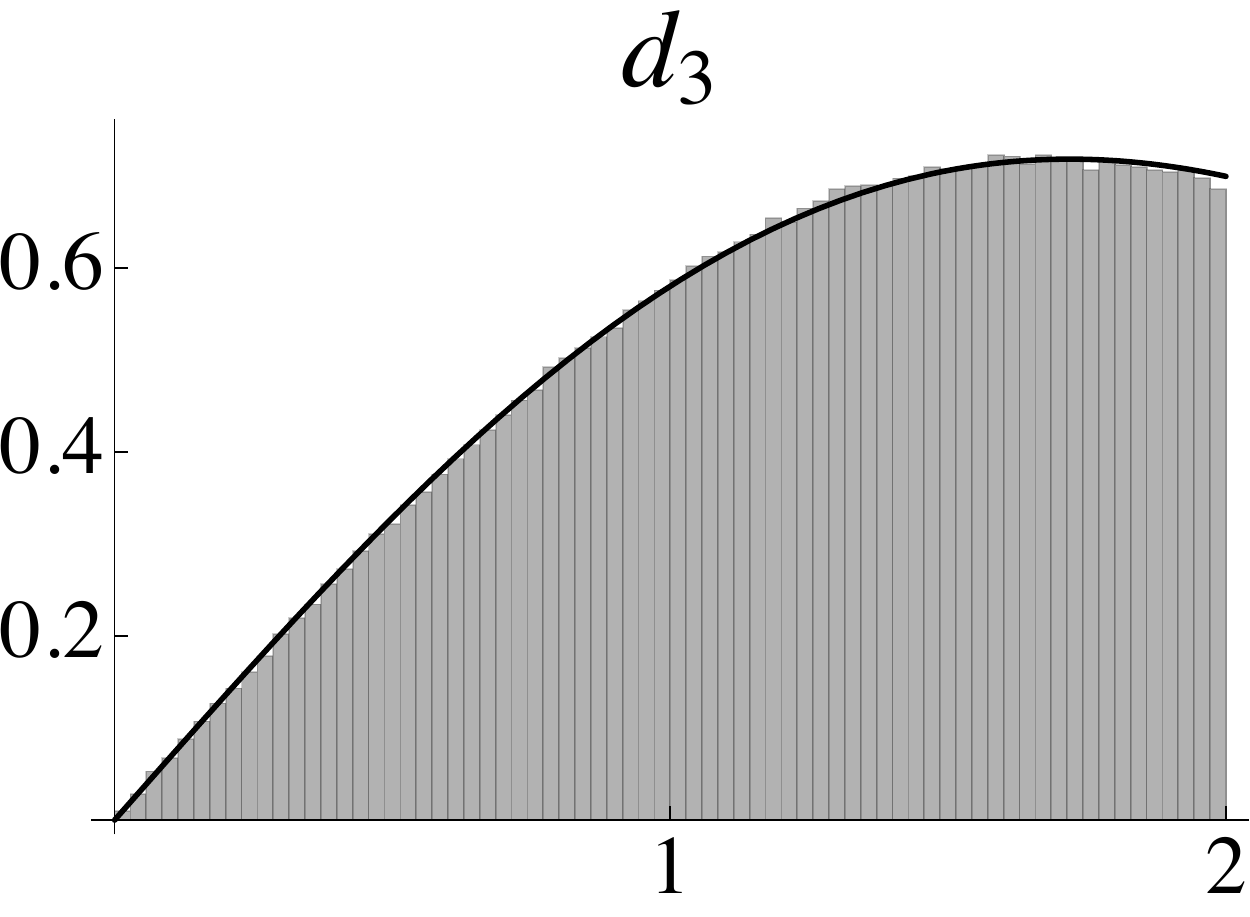}
		\includegraphics[width=1.4in]{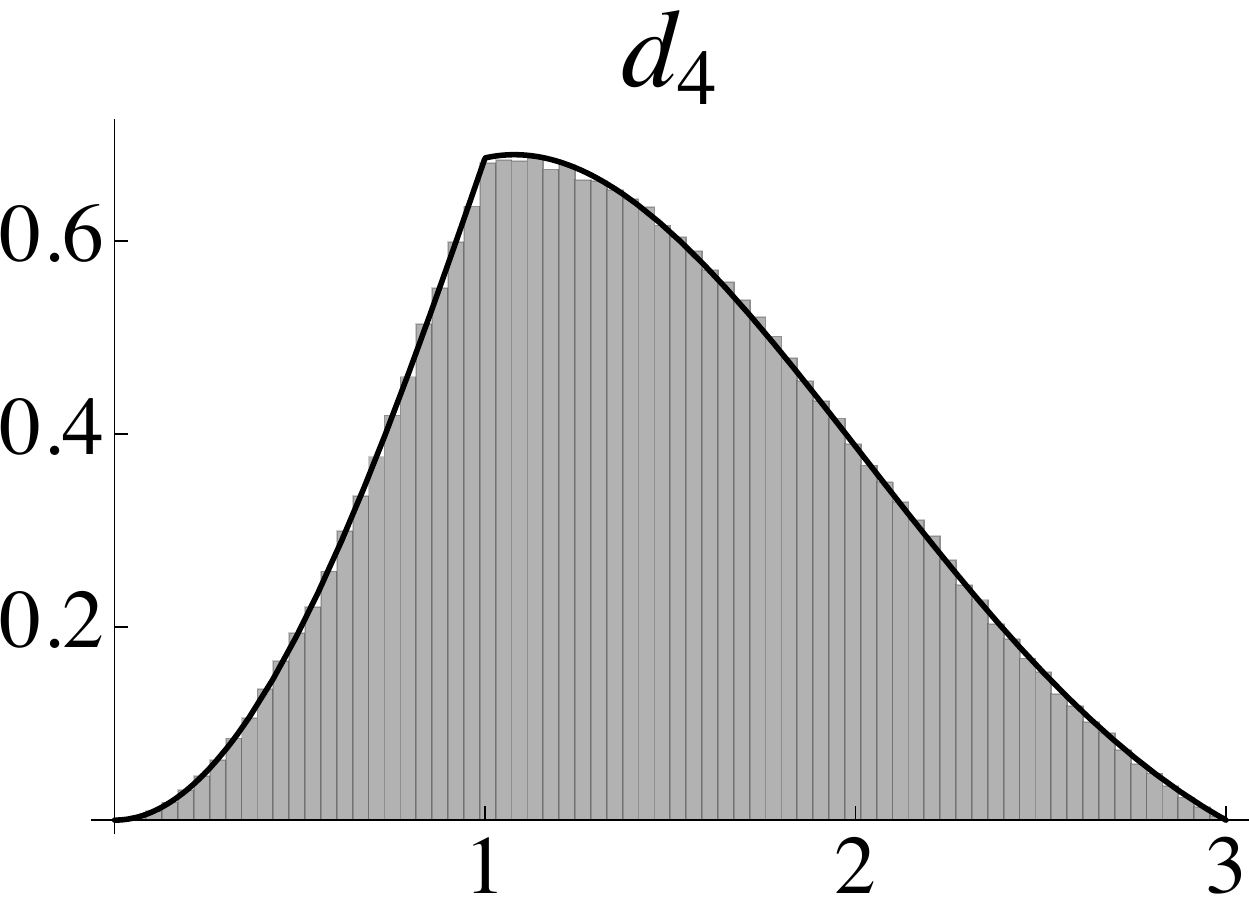}
		\includegraphics[width=1.4in]{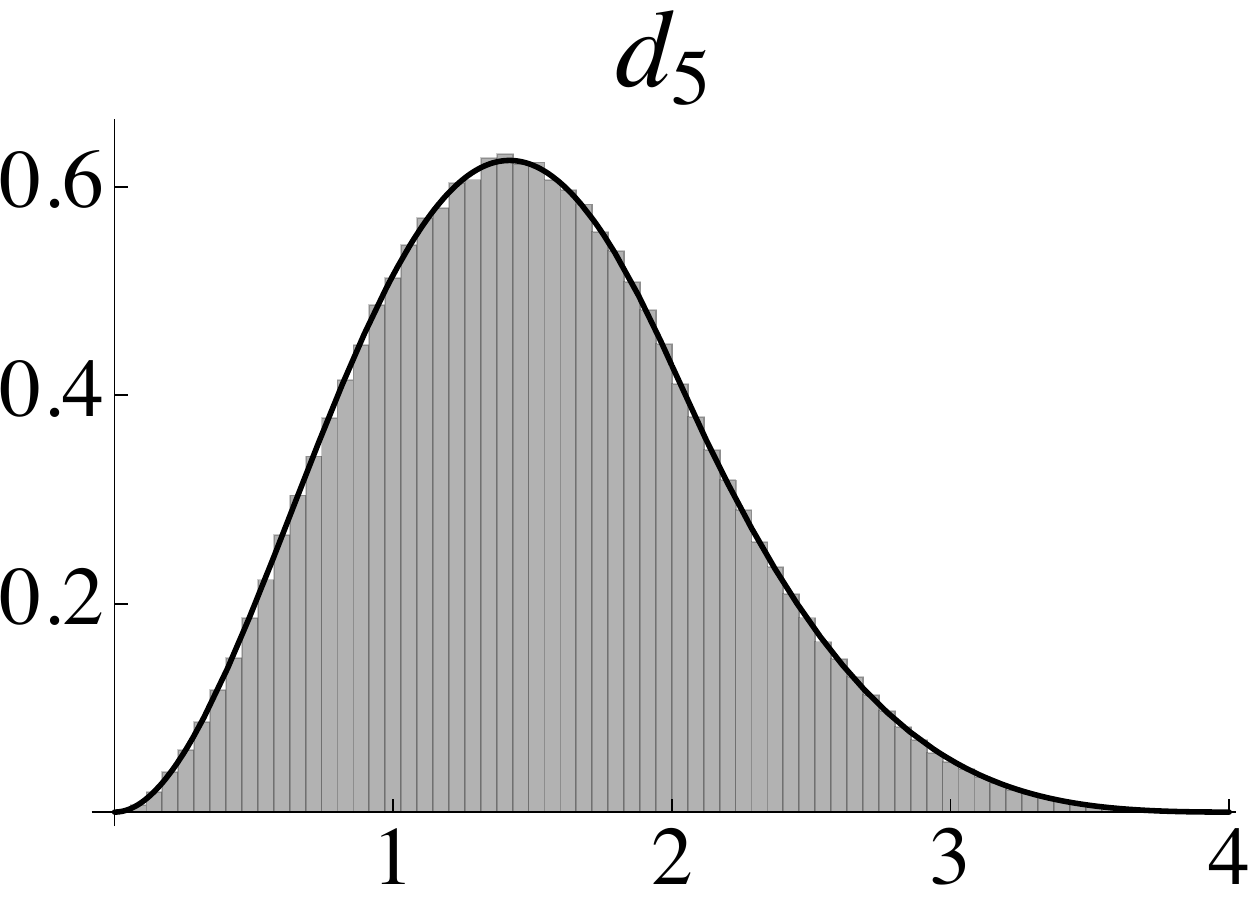}
		\includegraphics[width=1.4in]{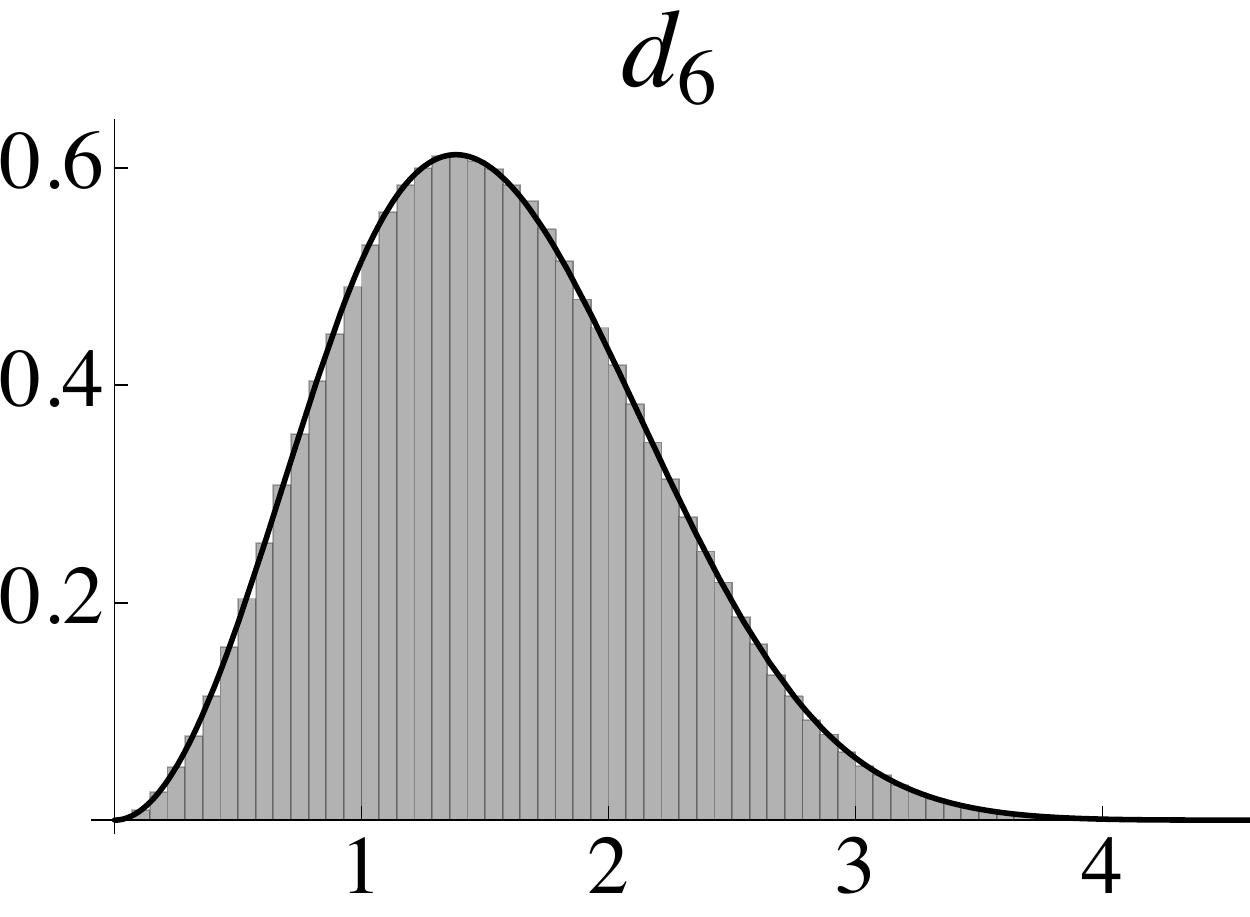}
		
		\vspace{.2in}
		\includegraphics[width=1.4in]{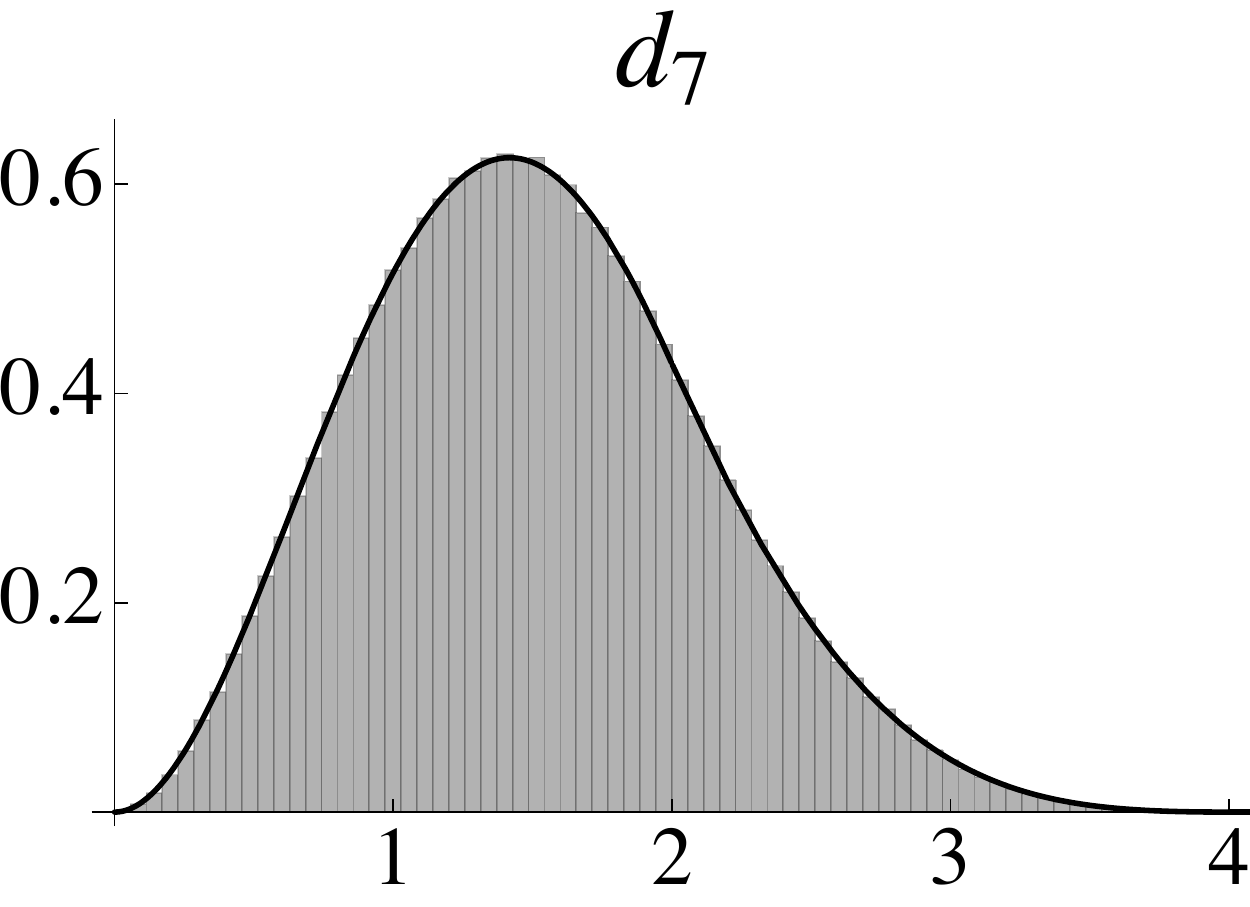}
		\includegraphics[width=1.4in]{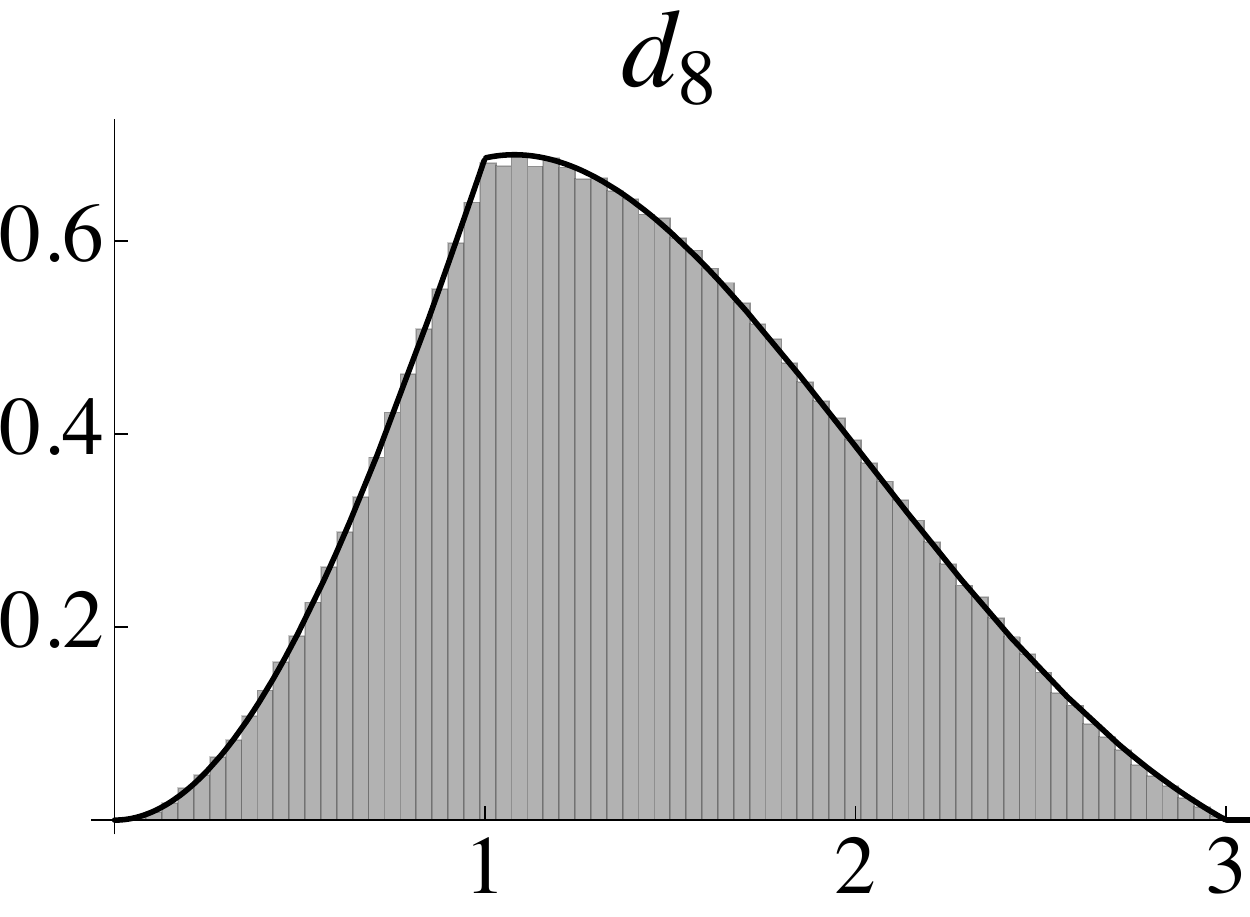}
		\includegraphics[width=1.4in]{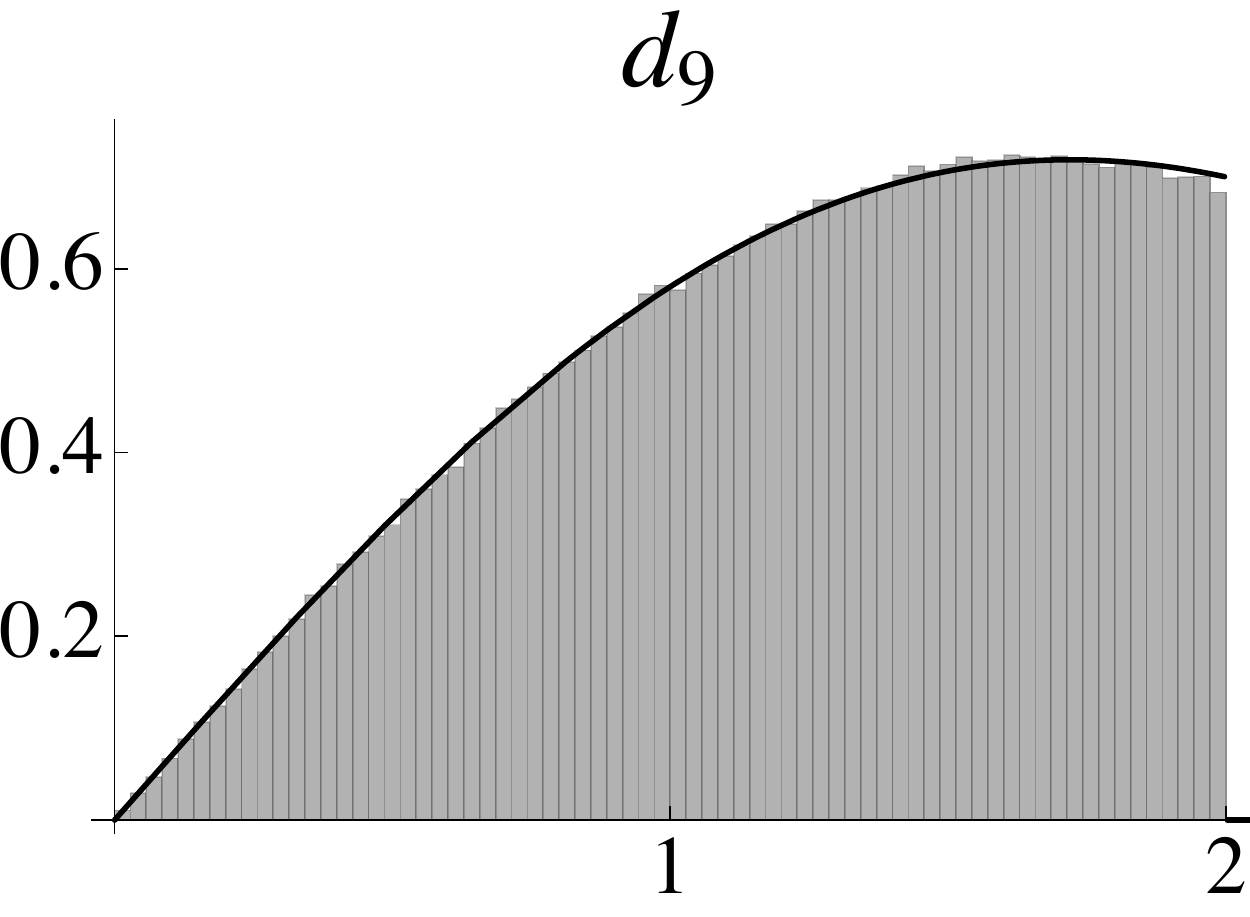}
	\caption{We generated 1,000,000 random equilateral 10-edge arms in $\R^3$. All 1,000,000 had geometric median closures, and these are the histograms of distances from the first vertex to the $i$th vertex in the resulting closed 10-gons, along with the chordlength densities for the standard measure on equilateral 10-gons.}
	\label{fig:10-gon chordlengths}
\end{figure}

\begin{conjecture}
	As $n \to \infty$, the pushforward measure from $\Arm(n,d,w)$ to $\Pol(n,d,w)$ converges to the standard measure.
	\label{conj:pushforward}
\end{conjecture}

Assuming the truth of this conjecture implies that, at least for large $n$, random elements of $\Pol(n,d,w)$ look essentially like geometric median closures of random elements of $\Arm(n,d,w)$. Since Corollary~\ref{cor:angular concentration} implies that individual edges are practically unchanged by closure, this would mean that \emph{all} local phenomena happen at essentially the same rate in $\Arm(n,d,w)$ and $\Pol(n,d,w)$.

When $d=3$, a particularly important local phenomenon is that of local knotting. Say that a subsegment $\varsigma$ of $\pmb{x} \in \Arm(n,3,w)$ is an \emph{$r$-local knot} if it only intersects the boundary of a ball $B$ of radius $r$ at its endpoints and $(B,\varsigma)$ forms a knotted ball-arc pair. Let $K^\text{Arm}(n,w,k,r)$ be the probability that a length-$k$ arc of a random element of $\Arm(n,3,w)$ is an $r$-local knot, and similarly for $K^\text{Pol}(n,w,k,r)$.

\begin{conjecture}
	For small $r$ and large $n$ and for $k \ll n$, $K^\text{Arm}(n,w,k,r)\simeq K^\text{Pol}(n,w,k,r)$.
	\label{conj:local knotting}
\end{conjecture}

\section*{Acknowledgements}
This paper is a contribution to the Festschrift for Stu Whittington, a giant in the area of random polymers and random knots. We are indebted to Stu for years of insightful talks, perceptive questions, and remarkable mathematical results. His interest, enthusiasm, and explanation of the importance of these questions to the polymer science community have shaped our mathematical trajectory more than we can say. 

We are also grateful for the continued support of the Simons Foundation (\#524120 to Cantarella, \#354225 to Shonkwiler), the German Research Foundation (DFG-Grant RE 3930/1--1, to Reiter), and the organizers of the ``Workshop on Topological Knots and Polymers'' at Ochanomizu University, where key steps in the present work came together. In particular, we are indebted to Cristian Micheletti, Tetsuo Deguchi, Rob Kusner (for reducing everything to conformal geometry yet again!), Eric Rawdon, and Erik Schreyer for many helpful conversations. As always, we look to Yuanan Diao for inspiration -- one of the motivations for this paper was the desire to find an alternate proof of~\cite{Diao:1995iw}.

\appendix*
\section{Proof of the hypergeometric formula for the expected distance to the sphere} \label{appendix}
Recall that $\Ed(\vec{y})$ is the expected distance from the point $\vec{y} \in \R^d$ to the unit sphere. Since it is spherically symmetric, $\Ed(\vec{y})$ depends only on $\|\vec{y}\|$.

\begin{proposition}\label{prop:mister ed}
	$\Ed(\vec{y})$ is given as a function of $r = \|\vec{y}\|$ by
	\[
	\Ed(r) = \, _2F_1 \left( -\frac{1}{2}, \frac{1-d}{2}; \frac{d}{2}; r^2 \right).
	\]
\end{proposition}

We can compute $\Ed(r)$ by evaluating at $\vec{y} = (0,\dots , 0,r)$:
\begin{multline} \label{eq:ed simplification 1}
	\Ed(r) = \Ed((0,\dots ,0,r))=\frac{1}{\Vol S^{d-1}} \int_{\hat{x} \in S^{d-1}} \|\hat{x} - (0,\dots ,0,r)\| \dVol_{S^{d-1}} \\
	= \frac{1}{\Vol S^{d-1}} \int_{\hat{x} \in S^{d-1}} \sqrt{1+r^2-2x_d r} \dVol_{S^{d-1}},
\end{multline}
where the integrand only depends on the last coordinate $x_d$ of the point on $S^{d-1}$. Then the formula for $\Ed$ will follow from a more general formula for functions on the sphere which only depend on a single coordinate:

\begin{lemma}\label{lem:integrating functions on the sphere}
	Suppose $\phi:S^{d-1} \to \R$ depends only on $x_i$; i.e. $\phi(\hat{x}) = \phi(x_i)$. Then
	\[
		\int_{\hat{x} \in S^{d-1}} \phi(\hat{x}) \dVol_{S^{d-1}} = \Vol S^{d-2} \int_{-1}^1 \phi(x_i)(1-x_i^2)^{\frac{d-3}{2}} \dx_i.
	\]
\end{lemma}

\begin{proof}
	Let $\pi_i: S^{d-1} \to \R$ be projection to the $i$th coordinate. Then the projection of $\nabla \pi_i$ to the tangent space of $S^{d-1}$ has norm $\sqrt{1-x_i^2}$, and hence the smooth coarea formula implies that
	\[
		\int_{\hat{x} \in S^{d-1}} \phi(\hat{x}) \dVol_{S^{d-1}} = \int_{x_i \in [-1,1]} \int_{\vec{y} \in \pi_i^{-1}(x_i)} \frac{\phi(\vec{y})}{\sqrt{1-x_i^2}} \dArea_{\pi_i^{-1}(x_i)} \dx_i.
	\]
	Since $\pi_i^{-1}(x_i)$ is a $(d-2)$-dimensional sphere of radius $\sqrt{1-x_i^2}$, which has area form $\left(1-x_i^2\right)^{\frac{d-2}{2}}\dArea_{S^{d-2}}$, and since $\phi$ is constant on each level set, the above reduces to
	\[
		\Vol S^{d-2} \int_{-1}^1 \phi(x_i)(1-x_i^2)^{\frac{d-3}{2}} \dx_i
	\]
	as desired.
\end{proof}

Combining this result with~\eqref{eq:ed simplification 1} shows that
\begin{equation}\label{eq:ed simplification 2}
	\Ed(r) = \frac{\Vol S^{d-2}}{\Vol S^{d-1}} \int_{-1}^1 \sqrt{1+r^2-2x_d r}(1-x_d^2)^{\frac{d-3}{2}} \dx_d,
\end{equation}
which is the starting point of our derivation of the hypergeometric formula.

\begin{proof}[Proof of Proposition~\ref{prop:mister ed}]
Using $\Vol S^k = \frac{2\pi^{\nicefrac{(k+1)}{2}}}{\Gamma(\nicefrac{(k+1)}{2})}$ and the gamma function duplication formula ${\Gamma(\zeta)\Gamma(\zeta+\frac{1}{2}) = 2^{1-2\zeta}\sqrt{\pi}\Gamma(2\zeta)}$, we can write the ratio of sphere volumes as
\[
	\frac{\Vol S^{d-2}}{\Vol S^{d-1}} = \frac{2^{2-d}\Gamma(d-1)}{\Gamma(\frac{d-1}{2})\Gamma(\frac{d-1}{2})}.
\]
Substituting this into~\eqref{eq:ed simplification 2} and completing the square inside the square root yields
\[
	(1+r) \frac{2^{2-d}\Gamma(d-1)}{\Gamma(\frac{d-1}{2})\Gamma(\frac{d-1}{2})} \int_{-1}^1 \left(1-\frac{2r}{(1+r)^2}(1+x_d)\right)^{\frac{1}{2}} (1-x_d)^{\frac{d-3}{2}}(1+x_d)^{\frac{d-3}{2}} \dx_d.
\]
Making the substitution $u=\frac{1+x_d}{2}$ produces
\[
	\Ed(r) = (1+r) \frac{2^{2-d}\Gamma(d-1)}{\Gamma(\frac{d-1}{2})\Gamma(\frac{d-1}{2})} \int_{-1}^1 u^{\frac{d-3}{2}}(1-u)^{\frac{d-3}{2}}\left(1-\frac{4r}{(1+r)^2}u\right)^{\frac{1}{2}} \dx_d,
\]
which is the standard integral representation of $(1+r)\, _2F_1\left(-\frac{1}{2}, \frac{d-1}{2}; d-1; \frac{4r}{(1+r)^2}\right)$. In turn, applying Kummer's quadratic transformation~\cite[9.134.3]{Gradshteyn:2015bj} yields the desired formula
\[
	\Ed(r) = \, _2F_1\left(-\frac{1}{2}, \frac{1-d}{2}; \frac{d}{2}; r^2\right).
\]
\end{proof}

\bibliography{openclosed_papers_extra,openclosed_papers}

\begin{thebibliography}{10}

\bibitem{Bhatia:2000ge}
Rajendra Bhatia and Chandler Davis.
\newblock {A better bound on the variance}.
\newblock {\em The American Mathematical Monthly}, 107(4):353--357, 2000.

\bibitem{Cantarella:2013wl}
Jason Cantarella and Clayton Shonkwiler.
\newblock {The symplectic geometry of closed equilateral random walks in
  3-space}.
\newblock {\em The Annals of Applied Probability}, 26(1):549--596, 2016.

\bibitem{Demaine:2007jh}
Erik~D Demaine and Joseph O'Rourke.
\newblock {\em {Geometric Folding Algorithms: Linkages, Origami, Polyhedra}}.
\newblock Cambridge University Press, Cambridge, 2007.

\bibitem{Diao:1995iw}
Yuanan Diao.
\newblock {The knotting of equilateral polygons in $\bf R^3$}.
\newblock {\em Journal of Knot Theory and its Ramifications}, 4(2):189--196,
  1995.

\bibitem{NIST:DLMF}
{\it NIST Digital Library of Mathematical Functions}.
\newblock http://dlmf.nist.gov/, release 1.0.18 of 2018--03--27.
\newblock F.~W.~J. Olver, A.~B. {Olde Daalhuis}, D.~W. Lozier, B.~I. Schneider,
  R.~F. Boisvert, C.~W. Clark, B.~R. Miller and B.~V. Saunders, eds.

\bibitem{Douady:1986go}
Adrien Douady and Clifford~J Earle.
\newblock {Conformally natural extension of homeomorphisms of the circle}.
\newblock {\em Acta Mathematica}, 157(1-2):23--48, 1986.

\bibitem{Hamacher:2002vp}
Zvi Drezner, Kathrin Klamroth, Anita Sch{\"o}bel, and George~O Wesolowsky.
\newblock {The Weber Problem}.
\newblock In Horst~W Hamacher and Zvi Drezner, editors, {\em Facility
  Location}, pages 1--36. Springer-Verlag, Berlin, 2002.

\bibitem{Dubhashi:2009ho}
Devdatt~P Dubhashi and Alessandro Panconesi.
\newblock {\em {Concentration of Measure for the Analysis of Randomized
  Algorithms}}.
\newblock Cambridge University Press, Cambridge, 2009.

\bibitem{FloryPaulJ1969Smoc}
Paul~J Flory.
\newblock {\em {Statistical Mechanics of Chain Molecules}}.
\newblock Interscience Publishers, New York, 1969.

\bibitem{Gradshteyn:2015bj}
Izrail~S Gradshteyn and Iosif~M Ryzhik.
\newblock {\em {Table of Integrals, Series, and Products}}.
\newblock Elsevier, Amsterdam, eighth edition, 2015.

\bibitem{Haberman:1989iq}
Shelby~J Haberman.
\newblock {Concavity and estimation}.
\newblock {\em The Annals of Statistics}, 17(4):1631--1661, 1989.

\bibitem{Horn:2013tf}
Roger~A Horn and Charles~R Johnson.
\newblock {\em {Matrix Analysis}}.
\newblock Cambridge University Press, Cambridge, second edition, 2013.

\bibitem{Howard:2008uy}
Benjamin Howard, Christopher Manon, and John~J Millson.
\newblock {The toric geometry of triangulated polygons in Euclidean space}.
\newblock {\em Canadian Journal of Mathematics}, 63(4):878--937, 2011.

\bibitem{hughes1995random}
Barry~D Hughes.
\newblock {\em {Random Walks and Random Environments: Volume 1: Random Walks}}.
\newblock Clarendon Press, Oxford, 1995.

\bibitem{Kapovich:1996p2605}
Michael Kapovich and John~J Millson.
\newblock {The symplectic geometry of polygons in Euclidean space}.
\newblock {\em Journal of Differential Geometry}, 44(3):479--513, 1996.

\bibitem{MR2004a:14059}
Michael Kapovich and John~J Millson.
\newblock {Universality theorems for configuration spaces of planar linkages}.
\newblock {\em Topology}, 41(6):1051--1107, 2002.

\bibitem{Khoi:2005ch}
Vu~The Khoi.
\newblock {On the symplectic volume of the moduli space of spherical and
  Euclidean polygons}.
\newblock {\em Kodai Mathematical Journal}, 28(1):199--208, 2005.

\bibitem{Anonymous:1UuVxm-1}
Harold~W Kuhn.
\newblock {A note on Fermat's problem}.
\newblock {\em Mathematical Programming}, 4(1):98--107, 1973.

\bibitem{Niemiro:1992ez}
Wojciech Niemiro.
\newblock {Asymptotics for $\mathrm{M}$-estimators defined by convex
  minimization}.
\newblock {\em The Annals of Statistics}, 20(3):1514--1533, 1992.

\bibitem{Rayleigh:1919do}
Lord Rayleigh.
\newblock {On the problem of random vibrations, and of random flights in one,
  two, or three dimensions}.
\newblock {\em Philosophical Magazine Series 5}, 37(220):321--347, 1919.

\bibitem{Rudnick:1987jn}
Joseph Rudnick and George Gaspari.
\newblock {The shapes of random walks}.
\newblock {\em Science}, 237(4813):384--389, 1987.

\bibitem{Tropp:2012fb}
Joel~A Tropp.
\newblock {User-friendly tail bounds for sums of random matrices}.
\newblock {\em Foundations of Computational Mathematics}, 12(4):389--434, 2012.

\end{thebibliography}

\end{document}